\documentclass[10pt, compsoc, journal, doublecolumn]{IEEEtran}
\usepackage{cite}
\pdfoutput=1
\usepackage{graphicx}
\usepackage[utf8]{inputenc}
\usepackage[english]{babel}
\usepackage{amsthm}
\usepackage{amsmath}
\usepackage{url}
\usepackage{bbm}
\usepackage[outdir=./]{epstopdf}
\usepackage[linesnumbered,ruled,vlined]{algorithm2e}
\usepackage{algpseudocode}
\algrenewcommand\algorithmicrequire{\textbf{Input:}}
\algrenewcommand\algorithmicensure{\textbf{Output:}}
\usepackage{threeparttable}
\usepackage{makecell}
\usepackage{amsmath,amssymb,amsfonts}
\usepackage{amsmath}
\usepackage{textcomp}
\usepackage{xcolor}

\usepackage{mathtools}

\usepackage[top=1in,bottom=0.75in,right=0.75in,left=0.75in]{geometry}
\usepackage{subfigure}
\def\BibTeX{{\rm B\kern-.05em{\sc i\kern-.025em b}\kern-.08em
    T\kern-.1667em\lower.7ex\hbox{E}\kern-.125emX}}

\newtheorem{theorem}{Theorem}
\newtheorem{corollary}{Corollary}[theorem]
\newtheorem{lemma}[theorem]{Lemma}
\newcommand{\equref}[1]{Eq.~(\ref{#1})}
\newcommand{\ineqref}[1]{Ineq.~(\ref{#1})}
\newcommand{\figref}[1]{Fig.~\ref{#1}}
\newcommand{\tabref}[1]{Table~\ref{#1}}
\newcommand{\algoref}[1]{Algorithm~\ref{#1}}
\newcommand{\secref}[1]{Section~\ref{#1}}
\newcommand{\theoref}[1]{Theorem~\ref{#1}}
\newcommand{\lemmaref}[1]{Lemma~\ref{#1}}

\SetKwInput{KwInput}{Input}
\SetKwInput{KwOutput}{Output}
\usepackage{etoolbox}
\ifCLASSINFOpdf
\else

\fi

\begin{document}

%
\title{On Batch-Processing Based Coded Computing for Heterogeneous Distributed Computing Systems}

\author{Baoqian~Wang,~\IEEEmembership{Student Member,~IEEE}, 
    Junfei~Xie,~\IEEEmembership{Member,~IEEE}, 
    Kejie~Lu,~\IEEEmembership{Senior Member,~IEEE}, 
    Yan~Wan,~\IEEEmembership{Senior Member,~IEEE},
    and Shengli~Fu,~\IEEEmembership{Senior Member,~IEEE}
\thanks{Baoqian Wang is with the Department of Electrical and Computer Engineering, San Diego State University, and University of California, San Diego, San Diego, CA, 92182 (e-mail: {\tt\small bwang4848@sdsu.edu)}.}
\thanks{Junfei Xie is with the Department of Electrical and Computer Engineering, San Diego State University, San Diego, CA, 92182 (e-mail: {\tt\small jxie4@sdsu.edu}). Corresponding author.} 
	\thanks{Kejie Lu is with the Department of Computer Science and Engineering, University of Puerto Rico at Mayag\"{u}ez, Mayag\"{u}ez, Puerto Rico, 00681, e-mail:  ({\tt\small kejie.lu@upr.edu}).	}
\thanks{Yan Wan is with the Department of Electrical Engineering, University of Texas at Arlington, Arlington, Texas, 76019 (e-mail:  {\tt\small yan.wan@uta.edu}).}
\thanks{Shengli Fu is with the Department of Electrical Engineering, University of North Texas, Denton, Texas, 76201 (e-mail: {\tt\small Shengli.Fu@unt.edu}).}
\thanks{Manuscript received October 1, 2020. 
}}

\IEEEtitleabstractindextext{\begin{abstract}
In recent years, \emph{coded distributed computing} (CDC) has attracted significant attention, because it can efficiently facilitate many delay-sensitive computation tasks against unexpected latencies in  distributed computing systems. Despite such a salient feature, many design challenges and opportunities remain. In this paper, we focus on practical computing systems with heterogeneous computing resources, and design a novel CDC approach, called \textit{batch-processing based coded computing} (BPCC), which exploits the fact that every computing node can obtain some coded results before it completes the whole task. To this end, we first describe the main idea of the BPCC framework, and then formulate an optimization problem for BPCC to minimize the task completion time by configuring the computation load.  Through formal theoretical analyses, extensive simulation studies, and comprehensive real experiments on the Amazon EC2 computing clusters, we demonstrate  promising performance of the proposed BPCC scheme, in terms of high computational efficiency and robustness to uncertain disturbances.
\end{abstract}

}

\maketitle

\section{Introduction}
In recent years, distributed computing has been widely adopted to perform various computation tasks in different computing systems \cite{kartik97task,hong04distributed, lu19toward}. For instance, to perform big data analytics in cloud computing systems, MapReduce \cite{dean08mapreduce} and Apache Spark \cite{zaharia10spark} are the two prevalent modern distributed computing frameworks that  process data in the order of petabytes. 

Despite the importance of distributed computing, many design challenges remain. One major challenge is that many computing frameworks are vulnerable to uncertain  disturbances, such as node/link failures, communication congestion, and slow-downs  \cite{dean13tail}. Such  disturbances, which can be modeled as  stragglers that are slow or even fail in returning results, have been observed in many large-scale computing systems such as cloud computing\cite{yang2019timely}, mobile edge computing\cite{kim2020coded}, and fog computing\cite{yue2020coding}.

A variety of  solutions have been developed in the literature to address  stragglers. For example, the authors of \cite{zaharia08improving} proposed to identify and blacklist nodes that are in bad health and to run tasks only on well-performed nodes. However, empirical studies show that stragglers can occur in non-blacklisted nodes as well \cite{dean12achieving,ananthanarayanan13effective}. As another type of solution, delayed computation tasks can be re-executed in a speculative manner \cite{dean08mapreduce, ananthanarayanan10reining, zaharia08improving, melnik10dremel}. Nevertheless, such speculative execution techniques have to wait to collect the performance statistics of the tasks before generating speculative copies and thus have limitations in dealing with small jobs \cite{ananthanarayanan13effective}. To avoid waiting and predicting stragglers, the authors of \cite{ananthanarayanan12why,ananthanarayanan13effective} suggested to execute multiple clones of each task and use results generated by the fastest clones. Although their results show the promising performance of this approach in reducing the average completion time of small jobs, the extra resources required for launching clones can be considerably large, because  multiple clones are executed for each task.

Instead of directly replicating the whole task, the \textit{coding} techniques can be adopted to introduce arbitrary redundancy into the computation in a systematic way. However, until a few years ago, the coding techniques have been mostly known for their capability in improving the resilience of communication, storage and cache systems to uncertain disturbances \cite{liu2020esetstore, li15coded, zhu2013speedup}. Lee \textit{et al.} \cite{lee16speeding,lee18speeding} presented the first \emph{coded distributed computing} (CDC) scheme to speed up matrix multiplication and data shuffling. Since then, CDC has attracted significant attention in the distributed computing community. Although a variety of CDC schemes have been developed to solve different computation problems, most of these schemes assume homogeneous computing nodes, which is not a common case in realistic scenarios. Moreover, they require each worker node to first complete the computation task and then send back the whole result to the master node, which introduces significant delays \cite{reisizadeh17coded, reisizadeh19coded, lee18speeding, lee16speeding}.

In this paper, we focus on a classical CDC task: matrix-vector multiplication and propose a novel coding scheme, called batch-processing based coded computing (BPCC), to speed up the computational efficiency of 
general distributed computing systems with heterogeneous computing nodes and improve their robustness to uncertain disturbances.
Unlike most existing CDC schemes, 
our BPCC allows each node to return partial computing results to the master node in batches before the whole computation task is completed. Therefore, BPCC 
achieves lower latency.  
Also worthy of note is that the partial results can be used to generate approximated solutions, e.g., by applying the singular value decomposition (SVD) approach in \cite{ferdinand16anytime}, which is very useful for  applications that require timely but unnecessarily optimized decisions such as emergency response. To the best of our knowledge, such a BPCC framework has not been fully investigated in the literature. 

This paper extends our earlier work presented in \cite{wang19coding}, and further makes the following new contributions.
\begin{enumerate}
\item \textit{An optimal load allocation strategy.}
For systems with heterogeneous computing nodes, equally distributing the computation load may lead to bad performance. To optimize the computational efficiency, we formulate an optimization problem for general BPCC with the assumption that the processing time of each computing node follows a shifted exponential distribution. To solve the optimization problem, we formulate alternative optimization problems, based on which we design an optimal load allocation scheme that  assigns proper amount of load to each node to achieve the minimal expected task completion time. 

\item \textit{Comprehensive theoretical analyses. 
}
We conduct formal theoretical analyses to prove the asymptotic optimality of BPCC and the impact of its important parameter.  We also prove that it outperforms a state-of-the-art CDC scheme for heterogeneous systems, called Heterogeneous Coded Matrix Multiplication (HCMM)\cite{reisizadeh17coded, reisizadeh19coded}. 

\item \textit{Extensive simulation and real experimental studies. }
To further demonstrate the performance of BPCC, we compare it with three benchmark schemes, including the Uniform Uncoded, Load-Balanced Uncoded, and HCMM. The simulation results show the impact of BPCC parameters including  number of batches and number of worker nodes. Specifically, the efficiency of BPCC improves with the increase of the number of batches and the solution of BPCC is optimal when the number of worker nodes approaches infinity. A sensitivity study shows the performance of BPCC when parameters in the computing model take erroneous values. 
Moreover, the simulation results also demonstrate that BPCC can improve computing performance by reducing the latency up to 73\%, 56\%, and 34\% over the aforementioned three benchmark schemes, respectively.
In the real experiments, we test all distributed computing schemes in the Amazon EC2 computing clusters. In particular, we deploy a heterogeneous computing cluster that consists of different machine instances in Amazon EC2. 
The results show that our BPCC scheme is more efficient and robust to uncertain disturbances than the  benchmark schemes.
\end{enumerate}


The rest of this paper is organized as follows.  \secref{sec:related} presents the related work. In \secref{sec:batch}, we introduce the system model and the BPCC framework, and then formulate an optimization problem for BPCC. To solve the optimization problem, in \secref{sec:optimalload}, we provide a two-step alternative formulation for which we design the BPCC scheme and conduct solid theoretical analysis to prove its optimality and understand the impact of its parameter. 
We then present extensive simulation and experimental results in \secref{sec:simulation} and \secref{sec:experiment}, respectively, before concluding the paper in \secref{sec:conclusion}. For
better readability, we move the proofs of all lemmas, theorems and corollaries to the Appendix.

\section{Related Work}
\label{sec:related}

Following the seminal work in \cite{li15coded,lee16speeding,lee18speeding}, many different computation problems have been explored using codes, such as the gradients \cite{tandon17gradient}, large matrix-matrix multiplication \cite{lee17high}, linear inverse problems \cite{yang17coded}, and nonlinear operations \cite{lee17coded}. Other relevant coded computation solutions include the ``Short-Dot'' coding scheme \cite{dutta17short-dot} that offers computation speed-up by introducing additional sparsity to the coded matrices and the unified coded framework \cite{li16unified,li18fundamental} that achieves the trade-off between communication load and computation latency. 

While most CDC schemes consider homogeneous computing nodes, there have been a few recent studies that investigated CDC over heterogeneous computing clusters. 
In particular, Kim \textit{et al.} \cite{kim19coded,kim19optimal} considered the matrix-vector multiplication problem and presented an optimal load allocation method that achieves a lower bound of the expected latency. Reisizadeh \textit{et al.} \cite{reisizadeh17coded} introduced a different approach, namely Heterogeneous Coded Matrix Multiplication (HCMM), that can maximize the expected computing results aggregated at the master node. In \cite{reisizadeh17coded, reisizadeh19coded}, the authors proved that the HCMM  is asymptotically optimal under the assumption that the processing time of each computing node follows a shifted exponential or Weibull distribution. Also of interest, Keshtkarjahromi \textit{et al.} \cite{keshtkarjahromi2018dynamic} considered the scenario when computing nodes have time-varying computing powers, and introduced a coded cooperative computation protocol that allocates tasks in a dynamic and adaptive manner. Narra \textit{et al.} \cite{narra2019slack} also developed an adaptive load allocation scheme and utilized a LSTM-based model to predict the computation capability of the worker nodes.

To reduce the output delay, 
there have been some attempts to enable early return of partial results \cite{ferdinand16anytime, ferdinand18hierarchical,mallick2019rateless}. In particular, an anytime coding technique was introduced in \cite{ferdinand16anytime}, which adopts the SVD to allow early output of approximated result. Also of interest is the study presented in \cite{ferdinand18hierarchical}, which introduced a hierarchical approach to address the limitations of above coding techniques in terms of wastefully ignoring the work completed by slow worker nodes. In particular, to better utilize the work completed by each worker node, it partitions the total computation at each worker node into layers of sub-computations, with each layer encoding part of the job. It then processes each layer sequentially. The final result can be obtained after the master node recovers all layers. The simulation results demonstrate the effectiveness of this approach in reducing the computation latency. However, as the worker nodes have to process the layers in the same order, the results obtained by slow worker nodes for layers that have already been recovered are useless. Furthermore, this approach, as well as aforementioned approaches, assumes homogeneous computing nodes. 
Another relevant study is presented in \cite{mallick2019rateless}, which introduced a rateless fountain coding scheme that can utilize partial results returned by  worker nodes. 


\section{System Models}
\label{sec:batch}

In this section, we first introduce the computing system for distributed matrix-vector multiplication. We then illustrate three computing schemes, including the proposed batch processing-based coded computing (BPCC). Finally, we formulate an optimization problem for BPCC.

\subsection{Computing System}
\label{sec:batch:system}

We consider a distributed computing system that consists of one master node and $N$ ($N \in \mathbb{Z}^+$) computing nodes, a.k.a., worker nodes. Using this system, we investigate how to quickly solve a matrix-vector multiplication problem, which is one of the most basic building blocks of many computation tasks. Specifically, we consider a matrix-vector multiplication problem $\boldsymbol{y}=\boldsymbol{A}\boldsymbol{x}$, where $\boldsymbol{y} \in \mathbb{R}^{r}$ is the output vector to be calculated, $\boldsymbol{x} \in \mathbb{R}^{m}$ is the input vector to be distributed from a master node to multiple workers, and $\boldsymbol{A} \in \mathbb{R}^{r \times m}$ is an $r \times m$ dimensional matrix pre-stored in the system. Both $r$ and $m$ can be very large, which implies that calculating $\boldsymbol{A} \boldsymbol{x}$ at a single computing node is not feasible. Finally, we define $[n] = \{1,2,\ldots,n\}$, where $n$ is an arbitrary positive integer, i.e., $n \in \mathbb{Z}^+$.

\subsection{Computing Schemes}

\subsubsection{Uncoded Distributed Computing} 

To solve the above problem, a traditional distributed computing scheme divides matrix $\boldsymbol{A}$ into a set of sub-matrices $\boldsymbol{A}_1, \boldsymbol{A}_2, \cdots, \boldsymbol{A}_N$, and pre-stores each sub-matrix $\boldsymbol{A}_i \in \mathbb{R}^{\ell_i\times m}$ in computing node $i$,  where $\forall i \in [N], \ell_i \in \mathbb{Z}^+$ and $\sum_{i=1}^N \ell_i = r$. Upon receiving the input vector $\boldsymbol{x}$, the master node sends vector $\boldsymbol{x}$ to all worker nodes. Each worker node $i$ then computes $\boldsymbol{y}_i = \boldsymbol{A}_i\boldsymbol{x}$ and returns the result to the master node. After all results are received, the master node aggregates the results and outputs $\boldsymbol{y} = [\boldsymbol{y}^{T}_1, \boldsymbol{y}^{T}_2, \cdots, \boldsymbol{y}^{T}_N]^{T}$, where $^{T}$ stands for transpose. 

Due to the existence of uncertain system disturbances, 
the uncoded computing scheme may defer or even fail the computation, because the delay or loss of any $\boldsymbol{y}_i$, $i \in [N]$, will affect the calculation of the final result $\boldsymbol{y} = \boldsymbol{A}\boldsymbol{x}$. To address this issue, more computing nodes can be used to perform distributed computing. For instance, the master node can have two or more computing nodes to compute $\boldsymbol{y}_i$. This approach, however, is not efficient because the cost can be unnecessarily large.

\subsubsection{Coded Distributed Computing (CDC)} 

In recent years, a more efficient computing paradigm, CDC, has been introduced to tackle the issue of uncertain disturbances. 
There are many CDC schemes in the literature,  and we consider a generic CDC scheme as follows. 

In this CDC scheme, $\boldsymbol{A}$ will first be used to calculate a larger matrix $\boldsymbol{\hat{A}}\in \mathbb{R}^{q\times m}$ 
with more rows, i.e., $q > r $, by using $
\boldsymbol{\hat{A}}=\boldsymbol{H} \boldsymbol{A}$, where $\boldsymbol{H} \in \mathbb{R}^{q\times r}$ is the encoding matrix with the property that any $r$ row vectors are linearly independent from each other \cite{lee17coded}. In other words, we can use any $r$ rows of $\boldsymbol{H}$ to create an $r \times r$ full-rank matrix. Note that this encoding procedure is performed offline and $\boldsymbol{\hat{A}}$ can be considered to be pre-stored in the system. Similar to the uncoded computing scheme, matrix $\boldsymbol{\hat{A}}$ can then be divided into $N$ sub-matrices $\boldsymbol{\hat{A}}_1, \boldsymbol{\hat{A}}_2, \cdots, \boldsymbol{\hat{A}}_N$, where $\boldsymbol{\hat{A}}_i \in \mathbb{R}^{\ell_i\times m}, \forall i \in [N]$,  $\sum_{i=1}^N \ell_i = q$, and each worker node $i$ 
calculates $\boldsymbol{\hat{y}}_{i} = \boldsymbol{\hat{A}}_{i} \boldsymbol{x}$. 

Different from the uncoded computing scheme, the master node does not need to wait for all worker nodes to complete their calculations, because it can recover $\boldsymbol{A} \boldsymbol{x}$ once the total number of rows of the received results is equal to or larger than $r$. In particular, suppose the master node receives $\boldsymbol{\hat{y}}_b \in \mathbb{R}^{r}$ at a certain time $t$, it can first infer that $\boldsymbol{\hat{y}}_b$ must satisfy
\begin{equation}
\boldsymbol{\hat{y}}_{b}  = \boldsymbol{\hat{H}}_b\boldsymbol{A}\boldsymbol{x},\nonumber
\end{equation} 
where $\boldsymbol{\hat{H}}_b \in \mathbb{R}^{r\times r}$ is a sub-matrix of the encoding matrix $\boldsymbol{H}$ corresponding to $\boldsymbol{\hat{y}}_{b}$. The master node can then calculate
\begin{equation}
\boldsymbol{y} = \boldsymbol{A}\boldsymbol{x} = \boldsymbol{\hat{H}}_{b}^{-1}\boldsymbol{\hat{y}}_b. \label{eq:eq1} 
\end{equation}

\subsubsection{BPCC}

In the literature, most existing CDC schemes assume that each worker node $i$  sends the complete $\boldsymbol{\hat{y}}_i$ to the master node when it is ready, which may incur large delays. To further speed up the computation, we propose a novel BPCC scheme and the main idea is to allow each worker node to return \textit{partial results} to the master node.

Specifically, we consider that each worker node $i$ equally divides the pre-stored encoded matrix $\boldsymbol{\hat{A}}_{i}$ row-wise into $p_i$ sub-matrices, named as \textit{batches}, where $p_i \in \mathbb{Z}^+$ is the number of batches and $p_i \leq \ell_i$. Except the last batch, each batch has $\lceil\frac{\ell_i}{p_i}\rceil=b_i$ rows.  After receiving the input vector $\boldsymbol{x}$ from the master node, the worker node multiplies each batch with $\boldsymbol{x}$ and will send back the partial results once available. Suppose that the master node receives $s_i(t)$ batches from the worker node $i$ by time $t$, where $0 \leq s_i(t) \leq p_i$, it can then recover the final result when $\sum_{i=1}^N \min(\ell_i, s_i(t) b_i)\geq r$, by using \equref{eq:eq1}.

\subsection{Problem Formulation}
\label{sec:problem_formulation}

In the previous sub-section, we introduced the key idea of the BPCC scheme. 
In the following study, we focus on  optimizing the performance of BPCC. Specifically, we  consider minimizing the task completion time. 
This is achieved by allocating proper computation load (i.e., $\ell_i$) to each worker node. 

We now define $T$ as the amount of time to complete a computation task. Given the number of batches for each worker node $\boldsymbol{p} = (p_1,p_2,\ldots,p_N)$, where $p_i \in \mathbb{Z}^+$, $\forall i\in[N]$, the optimization can be formulated as follows:
\begin{equation}
\label{eq:main_problem}
\begin{aligned} 
\mathcal{P}_{\mathrm{main}}: 
& \underset{\boldsymbol{\ell}} {\text { minimize }} 
& ~
&\mathbb{E}\left[T\right] \\
& \text { subject to } 
& 
& \ell_i\in  \mathbb{Z}^+, \forall i \in [N] \\
& 
& 
& \ell_i \geq p_i, \forall i \in [N]
\end{aligned}
\end{equation} 
where $\boldsymbol{\ell}=(\ell_1,\ell_2,\ldots,\ell_N)$. 


To facilitate further analysis, we assume that the computation task scales with $N$, i.e., $r = \Theta(N)$. Next, we assume that the computing nodes are fixed with time-invariant computation capabilities, and the network maintains a stable communication delay during the computing process. 

We now consider the behavior of  waiting time, which is defined as the duration from the time that the master node distributes $\boldsymbol{x}$ to the time that it receives a certain result. For BPCC, we let $T_{k,i}$ be the waiting time for the master node to receive $k$ batches from worker node $i$, $k\in \mathbb{Z}^+$. Clearly, $T_{k,i}$ can be modeled as a random variable following a certain probability distribution. Following the modeling technique used in recent studies \cite{reisizadeh19coded,lee16speeding,lee18speeding,ferdinand18hierarchical}, we consider that $T_{k,i}$ follows a shifted exponential distribution defined below:
\begin{equation}
\begin{gathered}
\operatorname{Pr}(T_{k,i} \leq t) = 
\begin{cases}
1-e^{-\mu_i(\frac{t}{kb_i}-\alpha_i)} & \mbox{if~} t \geq kb_i\alpha_i \\ 0 & \mbox{otherwise,} \end{cases}
\end{gathered} \label{eq:shift_exponential}
\end{equation}
where $\mu_i$ and $\alpha_i$ are straggling and shift parameters, respectively, and $\mu_i$ and $\alpha_i$ are \textbf{positive} constants for all $ i \in [N]$. Furthermore, we assume that $T_{k,i}$ is independent from $T_{k',j}$, $\forall j \in [N]$, $j\neq i$, $k'\in \mathbb{Z}^+$. 

Based on the above definitions and assumptions, we  see that $T$ must satisfy $\sum_{i=1}^Ns_i(T) b_i\geq r$. In the following sections, we will first discuss how to solve the optimization problem, in which we will conduct theoretical analysis to show the optimality and advantages of BPCC. We will then conduct extensive simulation and real experimental studies to validate the assumptions and to evaluate performance of the optimization algorithm.

\section{Main Results}
\label{sec:optimalload}

In this section,  we aim to solve the optimization problem $\mathcal{P}_{\mathrm{main}}$. In particular, we will first provide a simplified formulation, for which we then apply a two-step alternative formulation. Next, we show how to solve the alternative problems and prove the optimality of the solution. We then analyze the impact of parameter $p_i, \forall i \in [N]$ on the solution, and finally prove that this solution outperforms a recent CDC scheme without batch processing.

\subsection{Notations for Asymptotic Analysis} 

For any two given functions $f(n)$ and $g(n)$, $f(n)=\Theta(g(n))$ if and only if there exist positive constants $c_1$, $c_2$, and $n_0$ such that $0\leq c_1 g(n)\leq f(n) \leq c_2g(n)$ for all $n\geq n_0$; $f(n)=\mathcal{O}(g(n))$ if and only if there exist constants $n_0$ and $c$ such that $f(n)\leq cg(n)$ for all $n\geq n_0$; and $f(n)=o(g(n))$, if and only if $\lim_{n \rightarrow \infty}\frac{f(n)}{g(n)}=0$. 

\subsection{A Simplified Formulation}

We relax the constraint from $\ell_i \in \mathbb{Z}^+$ to $\ell_i \geq 0$, $\forall i \in [N]$ to simplify the analysis. We also remove the constraint $\ell_i \geq p_i$, $\forall i \in [N]$, by assuming that $p_i \in \mathbb{Z}^+$ is properly selected such that the optimal solution satisfies this constraint.  
Consequently, the problem in \equref{eq:main_problem} can be formulated as follows: 
\begin{equation}
\begin{aligned} 
\mathcal{P'}_{\mathrm{main}}: 
& \underset{\boldsymbol{\ell}}{\text { minimize }} 
& ~
& \mathbb{E}\left[T\right]\\
& \text { subject to } 
& ~
& \ell_i \geq 0, \forall i \in [N],
\end{aligned}\nonumber
\end{equation}
Once the above problem is solved, we can round each optimal load number $\ell_i$ up to its nearest integer using the ceiling function (denoted as $\lceil \ \rceil$). Note that the effect of this rounding step is negligible in practical applications with large load numbers, such as those considered in our simulation and experimental studies\cite{reisizadeh19coded}. In cases when the derived load number $\ell_i$ is smaller than $p_i$, we reduce the value of $p_i$ until this assumption holds. Note that we can always find such $p_i$ that satisfies the constraint, as the derived load number $\ell_i$ is always larger than or equal to 1.

\subsection{A Two-Step Alternative Formulation}

To solve the above problem, which is NP-Hard, we provide a two-step alternative formulation, 
inspired by \cite{reisizadeh19coded}. We will show later that this alternative formulation provides an asymptotically optimal solution to problem $\mathcal{P'}_{\mathrm{main}}$. 

The key idea of the two-step alternative formulation is to first maximize the amount of results accumulated at the master node by a feasible time $t$, i.e., $t \geq \max_i\{\alpha_i\ell_i\}$, and then minimize time $t$ such that sufficient amount of results are available to recover the final result. 
In particular, we let $S(t) = \sum_{i=1}^N s_i(t)b_i$ be the amount of results received by the master node by time $t$, where $b_i = \frac{\ell_i}{p_i}$ is the batch size. For a feasible time $t$, we first maximize the expected amount of results received by the master node, through solving the following problem: 
\begin{equation}
\begin{aligned} 
\mathcal{P}_{\mathrm{alt}}^{(1)}: 
& \underset{\boldsymbol{\ell}}{\text { maximize }} 
& ~
&\mathbb{E}\left[S(t)\right]\\
& \text { subject to } 
& 
& \ell_i \geq 0, \forall i \in [N]
\end{aligned}\nonumber
\end{equation}
After obtaining the solution to $\mathcal{P}_{\mathrm{alt}}^{(1)}$, denoted as 
$\boldsymbol{\ell}^{*}(t)=\left(\ell_{1}^{*}(t), \cdots, \ell_{N}^{*}(t)\right)$, 
we then minimize  the time $t$ such that there is a high probability that the results received by the master node by time $t$ are sufficient to recover the final result, by solving
\begin{equation}
\label{eq:eqal2}
\begin{aligned} \mathcal{P}_{\mathrm{alt}}^{(2)}: & \text { minimize } & &t \\ & \text { subject to } & & \operatorname{Pr}\left[S^{*}(t)<r\right]=o\left(\frac{1}{N}\right) \nonumber \end{aligned}
\end{equation}  
where $S^{*}(t)$ is the amount of results received by the master node by time $t$ for load allocation $\boldsymbol{\ell}^{*}(t)$. 

\subsection{Solution to the Two-Step Alternative Problem}

To solve the two-step alternative problem, we first consider $\mathcal{P}_{\mathrm{alt}}^{(1)}$. Note that, the expected amount of results received by the master node by time $t$ is:
\begin{align}
\label{eq:eqest_1}
\mathbb{E}[S(t)]&=\sum_{i=1}^{N}\mathbb{E}[s_{i}(t)b_i]\nonumber\\
&=  \sum_{i=1}^Nb_i\left[\sum_{k=1}^{p_i}k\operatorname{Pr}[s_i(t) = k] \right]
\end{align}
where $s_i(t)$ is an integer in range $0 \leq s_i(t) \leq p_i$, and $\operatorname{Pr}[s_i(t) = k]$ is the probability that the master node receives exactly $k$ batches from worker node $i$, 
\begin{align}
& \operatorname{Pr}[s_i(t) = k] \nonumber \\  
= &\begin{cases} 
    1 - \operatorname{Pr}(T_{1,i} \leq t),
    & k=0 \\
    \operatorname{Pr}(T_{k,i} \leq t) - \operatorname{Pr}(T_{k+1,i} \leq t),
    & 0 < k < p_i \\
    \operatorname{Pr}(T_{p_i,i} \leq t).
    & k = p_i 
\end{cases}\nonumber
\end{align}
$\mathbb{E}[S(t)]$ in \equref{eq:eqest_1} can then be computed by:
\begin{align}
\mathbb{E}[S(t)]
&=\sum_{i=1}^Nb_i\left[\sum_{k=1}^{p_i-1}k\operatorname{Pr}[s_i(t) = k] 
+ p_i \operatorname{Pr}[s_i(t) = p_i] \right]\nonumber\\
&=\sum_{i=1}^{N}\sum_{k=1}^{p_i}b_i\operatorname{Pr}(T_{k,i} \leq t)\nonumber\\
& = \sum_{i=1}^{N}\sum_{k=1}^{p_i}b_i\left(1-e^{-\mu_i(\frac{t}{kb_i}-\alpha_i)}\right)\label{eq:eqest}\nonumber\\
& =\sum_{i=1}^N \left(\ell_i-b_i\sum_{k=1}^{p_i}e^{-\mu_i(\frac{t}{kb_i}-\alpha_i)}\right)\nonumber\\
& =\sum_{i=1}^N \left(\ell_i-\frac{\ell_i}{p_i}\sum_{k=1}^{p_i}e^{-\mu_i(\frac{tp_i}{k\ell_i}-\alpha_i)}\right) 
\end{align}

The solution to $\mathcal{P}_{\text { alt}}^{(1)}$ can then be obtained by solving the following equation for each $i \in [N]$: 
\begin{equation}
\begin{gathered}
    \frac{\partial}{\partial \ell_{i}} \mathbb{E}\left[S(t)\right]=1-\left[\sum_{k=1}^{p_i}\left(\frac{1}{p_i}+\frac{\mu_it}{\ell_ik}\right)e^{-\mu_i(\frac{tp_i}{k\ell_i}-\alpha_i)}\right]=0, \nonumber
\end{gathered}
\end{equation}
which yields:
\begin{equation}
\label{eq:eq67}
   \ell_{i}^{*}(t)=\frac{t}{\lambda_{i}}
\end{equation}
$\lambda_{i}$ is the positive solution to the following equation:
\begin{equation} \label{eq:lambda}
 \sum_{k=1}^{p_i}\left(\frac{1}{p_i}+\frac{\mu_i\lambda_i}{k}\right)e^{-\mu_i(\frac{\lambda_ip_i}{k}-\alpha_i)}=1,
\end{equation}
which is a constant independent of $t$. To show that \equref{eq:lambda} has a single positive solution, we can define an auxiliary function $f_i$ for each $i$:
\begin{equation}
f_i(x) = \sum_{k=1}^{p_i} \left(\frac{1}{p_i}+\frac{\mu_i x}{k}\right)
e^{-\mu_i(\frac{x p_i}{k}-\alpha_i)}.\nonumber
\end{equation}
We can see that $f_i(x)$ decreases monotonically with the increase of $x$ when $x > 0$. We can also find that $f_i(0) = e^{\mu_i \alpha_i} > 1$ and $f_i(\infty) = 0$. Based on these statements, we know that a unique $\lambda_i$ exists and can be efficiently solved using a numerical approach. Next, we show in \lemmaref{lemma:lambda} that $\lambda_i$ has closed-form infimum and supremum. 


\begin{lemma}
\label{lemma:lambda}
Let $\lambda_i$, $i\in [N]$, be the positive solution to \equref{eq:lambda}. Its infimum is given by
\begin{equation}
    \inf \lambda_i = \lim_{p_i \rightarrow \infty} \lambda_i=  \alpha_i, 
\end{equation}
In addition, its supremum is given by
\begin{equation}
\label{eq:lambda_sup}
  \sup \lambda_i  =  \frac{W(-e^{-\alpha_i\mu_i-1})+1}{-\mu_i},  
\end{equation}
which is attained when $p_i = 1$ and $W(\cdot)$ is the Lambert W function\cite{corless1996lambertw}. 
\end{lemma}

From \lemmaref{lemma:lambda}, we can derive that the condition $t \geq \max_i\{\alpha_i\ell_i(t)\}$ holds, as $t = \ell_{i}^{*}(t)\lambda_{i} \geq \ell_{i}^{*}(t) \alpha_{i}$ for each work node $i$. 

Next, we solve $\mathcal{P}_{\mathrm{alt}}^{(2)}$. Since this problem is also NP-hard, we here provide an approximated solution. In particular, we approximate its optimal solution, denoted as $t^*$, with value $\tau^*$, such that the expected amount of results accumulated at the master node by time $\tau^*$ equals to the amount of results required for recovering the final result, i.e., $\mathbb{E}[S^*(\tau^*)] = r$. To find the value of $\tau^*$, we let 
\begin{equation}
  \mathbb{E}[S^{*}(t)]=r.  \label{eq:alter2}
\end{equation}
Then, using the load allocation $ \ell_{i}^{*}(t)$ in \equref{eq:eq67}, the expected amount of results received by the master node is:
\begin{equation}
\begin{aligned}
\mathbb{E}[S^*(t)]
& =\sum_{i=1}^N \left(\ell^*_i(t)-\frac{\ell^*_i(t)}{p_i}\sum_{k=1}^{p_i}e^{-\mu_i(\frac{tp_i}{k\ell^*_i(t)}-\alpha_i)}\right)\label{eq:eqesstar}\\
& = \sum_{i=1}^{N}\frac{t}{\lambda_{i}}\left(1-\frac{1}{p_i}\sum_{k=1}^{p_i}e^{-\mu_i(\frac{\lambda_ip_i}{k}-\alpha_i)}\right).
\end{aligned}
\end{equation}

We can then find the solution to \equref{eq:alter2} as follows:
\begin{equation}
\label{eq:tau}
\begin{gathered}
\tau^{*}=\frac{r}{\beta}
\end{gathered}
\end{equation}
where 
\begin{equation} \label{eq:beta}
\displaystyle \beta = \sum_{i=1}^{N}\frac{1}{\lambda_{i}}\left(1-\frac{1}{p_i}\sum_{k=1}^{p_i}e^{-\mu_i(\frac{\lambda_ip_i}{k}-\alpha_i)}\right),
\end{equation}
which is also a constant.

Combining the solutions to $\mathcal{P}_{\mathrm{alt}}^{(1)}$ and $\mathcal{P}_{\mathrm{alt}}^{(2)}$, we can then derive the load allocation: 
\begin{equation} \label{eq:prob1_solution}
    \ell_{i}^{*}\left(\tau^{*}\right)=\frac{r}{\beta\lambda_{i}} 
\end{equation}
The procedures of BPCC are summarized in \algoref{alg:bpcc1}.
\begin{algorithm}[t] 
    \DontPrintSemicolon
    \KwInput{$r, N, \boldsymbol{p}=\{p_1,\ldots,p_N\}, \boldsymbol{\mu}=\{\mu_1,\ldots,\mu_N\}, \boldsymbol{\alpha}=\{\alpha_1,\ldots,\alpha_N\}$}
    \KwOutput{$\boldsymbol{\ell}$}
    \For{$i=1:N$}
    {
        Calculate $\lambda_i$ by solving \equref{eq:lambda}
    }
    Calculate $\beta$ by using \equref{eq:beta}\\
    \For{$i=1:N$}
    {
        Calculate $\ell_{i}^{*}$ by using \equref{eq:prob1_solution}
    }
    \textbf{Return} $\boldsymbol{\ell}=\{ \lfloor \ell_{1}^{*} \rceil, \lfloor \ell_{2}^{*} \rceil, \cdots, \lfloor \ell_{N}^{*} \rceil \}$
\caption{BPCC} \label{alg:bpcc1}
\end{algorithm}

\subsection{Optimality Analysis}

In this sub-section, we conduct theoretical analysis to investigate the performance of BPCC. 
Specifically, we first show in Lemma \ref{lemma:lm1}  the optimality of the approximated solution $\tau^*$ to $\mathcal{P}_{\text { alt}}^{(2)}$. We then show in   \theoref{thm:1} that the solution provided by BPCC 
is asymptotically optimal. Finally, we show in Theorem \ref{thm:tau} the accuracy of $\tau^*$ in approximating the expected execution time of BPCC.

\begin{lemma}
\label{lemma:lm1}
Let $t^*$ be the optimal solution to $\mathcal{P}_{\text { alt}}^{(2)}$, and $\tau^*$ be the approximated solution given by \equref{eq:tau}. If the batch processing time follows the shifted exponential distribution in \equref{eq:shift_exponential} and $r = \Theta(N)$, then 
\begin{equation}
\label{eq:lemma1}
    \tau^{*}-o(1) < t^{*} \leq \tau^{*}+o(1).
\end{equation}
\end{lemma}

Based on \lemmaref{lemma:lm1}, we next  show the asymptotic optimality of BPCC in \theoref{thm:1}.

\begin{theorem}
\label{thm:1}
Consider problem $\mathcal{P'}_\mathrm{main}$ with the batch processing time following the shifted exponential distribution in \equref{eq:shift_exponential} and $r = \Theta(N)$. Let $\mathbb{E}[T_{\mathrm{BPCC}}]$ and $\mathbb{E}[T_{\mathrm{OPT}}]$ be the expected execution time of BPCC and the optimal value of $\mathcal{P'}_\mathrm{main}$, respectively. The BPCC is asymptotically optimal, i.e.,
\begin{equation}
    \lim _{N \rightarrow \infty} \mathbb{E}\left[T_{\mathrm{BPCC}}\right]=\lim _{N \rightarrow \infty} \mathbb{E}\left[T_{\mathrm{OPT}}\right]
    \label{eq:thm2}
\end{equation}
\end{theorem}

Theorem \ref{thm:1} and Lemma \ref{lemma:lm1} further lead to the following theorem.
\begin{theorem}
\label{thm:tau}
Let $\tau^*$ be the approximated solution given by \equref{eq:tau} and $\mathbb{E}[T_{\mathrm{BPCC}}]$ be the expected execution time of BPCC. If the batch processing time follows the shifted exponential distribution in \equref{eq:shift_exponential} and $r = \Theta(N)$, then 
\begin{equation}
\tau^{*} =\lim _{N \rightarrow \infty} \mathbb{E}\left[T_{\mathrm{BPCC}}\right]
\end{equation}
\end{theorem}

\subsection{Analysis of the Impact of Parameter $\boldsymbol{p}$}
\label{sec:impact_of_p}
In the BPCC scheme shown in Algorithm \ref{alg:bpcc1}, we note that $\boldsymbol{p}$ is the only parameter that can be tuned, while the other parameters, including $r$, $N$, $\boldsymbol{u}$ and $\boldsymbol{\alpha}$, are determined by the specific computation task and  properties of the distributed computing system. In this sub-section, we analyze the impact of this important parameter $\boldsymbol{p}$ on the performance of BPCC in Theorem \ref{thm:convergence}. We then show in Theorem \ref{thm:4} that the approximated execution time of BPCC, i.e., $\tau^*$ given by \equref{eq:tau}, has closed-form infimum and supremum.




\begin{theorem}
\label{thm:convergence}
Consider problem $\mathcal{P'}_\mathrm{main}$ with the batch processing time following the shifted exponential distribution in \equref{eq:shift_exponential} and $r = \Theta(N)$. Let $\tau^*$ be the approximated execution time of BPCC given by \equref{eq:tau}. Then the increase of any $p_i$, $i\in [N]$, will cause $\tau^*$ to decrease. 
\end{theorem}


\begin{theorem}
\label{thm:4}
Consider problem $\mathcal{P'}_\mathrm{main}$ with the batch processing time following the shifted exponential distribution in \equref{eq:shift_exponential} and $r = \Theta(N)$. Let $\tau^*$ be the approximated execution time of BPCC given by \equref{eq:tau}. Then
\begin{align}
\label{eq:tau_convergence}
    \inf \tau^* = & \lim_{p_i \rightarrow \infty, \forall i\in [N]} \tau^* \nonumber \\
    = & \frac{r}{\sum_{i=1}^N\frac{1}{\alpha_i}(1-e^{\mu_i\alpha_i}\int_{0}^{1}e^{-\frac{\mu_i\alpha_i}{x}}dx)}, 
\end{align}
and
\begin{align}
    \sup \tau^* = &\max \tau^* \nonumber \\ = & \sum_{i=1}^{N}\frac{1}{\sup \lambda_{i}}\left(1-e^{-\mu_i(\sup \lambda_i-\alpha_i)}\right),
\end{align}
which is attained when $p_i = 1$, $\forall i \in [N]$. Here $\sup \lambda_i$ is given by \equref{eq:lambda_sup}. 
\end{theorem}

From \theoref{thm:4} and \equref{eq:prob1_solution}, we can derive the following corollary. 
\begin{corollary}
\label{corollary:2}
Consider problem $\mathcal{P'}_\mathrm{main}$ with the batch processing time following the shifted exponential distribution in \equref{eq:shift_exponential} and $r = \Theta(N)$. Let  $\ell_i^*$ be the solution of BPCC given by \equref{eq:prob1_solution}. Then when the approximated execution time  $\tau^*$ of BPCC given by \equref{eq:tau}  converges to its infimum,
$\ell^*_i$ converges to $\hat{\ell}_i$, where 
\begin{equation}
\label{eq:l_convergence}
\hat{\ell}_i  = 
\frac{r}{\alpha_i\sum_{j=1}^N\frac{1}{\alpha_j}(1-e^{\mu_j\alpha_j}\int_{0}^{1}e^{-\frac{\mu_j\alpha_j}{x}}dx)}.
\end{equation}
\end{corollary}
\subsection{Comparison with HCMM}

In this sub-section, we compare the performance of BPCC with HCMM \cite{reisizadeh19coded}, a state-of-the-art CDC scheme for heterogeneous worker nodes, and   show that BPCC outperforms HCMM in computational efficiency. 

HCMM can be considered as a special case of BPCC with $p_i = 1$, $\forall i \in [N]$. It assigns each worker node $i$ with load $\ell_{H,i} = \frac{r}{\beta_H \lambda_{H,i}}$, where $\lambda_{H,i}$ is the positive solution to $e^{\mu_i\lambda_{H,i}}=e^{\alpha_i\mu_i}(\mu_i\lambda_{H,i}+1)$ and $\beta_H = \sum_{i=1}^{N}\frac{\mu_i}{1+\mu_i\lambda_{H,i}}$. 
\theoref{thm:3} shows that BPCC is more efficient than HCMM.


\begin{theorem}
\label{thm:3}
Consider problem $\mathcal{P'}_\mathrm{main}$, with the batch processing time following a shifted exponential distribution in \equref{eq:shift_exponential} and $r = \Theta(N)$. Let $T_{\mathrm{BPCC}}$ and $T_{\mathrm{HCMM}}$ be the execution times of BPCC and HCMM, respectively. Then,
\begin{equation}
    \lim _{N \rightarrow \infty} \mathbb{E}\left[T_{\mathrm{BPCC}}\right]\leq\lim _{N \rightarrow \infty} \mathbb{E}\left[T_{\mathrm{HCMM}}\right] \nonumber
\end{equation}
\end{theorem}

\section{Simulation Studies}
\label{sec:simulation}

In this section, we conduct  simulation studies to evaluate the performance of the proposed BPCC scheme. Specifically, we first explain the simulation settings, including the distributed computing schemes and scenarios. We then elaborate on the impact of important parameters, including $p_i$, $N$, $\mu_i$ and $\alpha_i$, on the performance of the BPCC scheme. Finally, we compare the proposed BPCC scheme with benchmark schemes, including the state-of-the-art HCMM scheme \cite{reisizadeh19coded}.

\subsection{Simulation Settings}
\label{sec:simsettings}

\subsubsection{Distributed Computing Schemes}

In this study, we consider four distributed computing schemes:

\begin{itemize}
\item \textbf{Uniform Uncoded:} This method divides the computation loads equally, i.e., $\ell_i = \frac{r}{N}$, $\forall i \in [N]$.

\item \textbf {Load-Balanced Uncoded \cite{reisizadeh19coded}:} This method divides the computation loads according to the computing capabilities of the worker nodes. In particular, the computation load assigned to each worker node $i$ is inversely proportional to the expected time for this node to compute an inner product, i.e.,  $\ell_i \propto (\frac{\mu_i}{\mu_i \alpha_i +1})$ and $\sum_{i=1}^N \ell_i = r$. 

\item \textbf{HCMM \cite{reisizadeh19coded}:} In this method, the load assignment method in \cite{reisizadeh19coded} is used. Note that this is a special case of \algoref{alg:bpcc1}, in which $p_i=1, \forall i \in [N]$. The HCMM and BPCC have the exactly same load allocation for each worker. 

\item \textbf{BPCC:} In this scheme, \algoref{alg:bpcc1} is used, where $\boldsymbol{p}$ are the parameters to configure.
\end{itemize}

\subsubsection{Computation Scenarios}
\label{sec:computation_scenario}
To evaluate the performance of different distributed computing schemes, we consider the following four computation scenarios:

\begin{itemize}

\item \textbf{Scenario 1:} $r = 1 \times 10^4$ and $N = 10$.

\item \textbf{Scenario 2:} $r = 2 \times 10^4$ and $N = 10$.

\item \textbf{Scenario 3:} $r = 1 \times 10^4$ and $N = 20$.

\item \textbf{Scenario 4:} $r = 2 \times 10^4$ and $N = 20$.

\end{itemize}
\subsubsection{Simulation Method}

In our simulation, we implement all the aforementioned distributed computing schemes in MATLAB. We assume that 
the processing time of each node follows the shifted exponential distribution in \equref{eq:shift_exponential}. Specifically, for each experiment of a scenario, we choose the straggling parameters $\mu_i, \forall i \in [N]$ randomly in $[1,50]$, and calculate each shift parameter $\alpha_i=\frac{1}{\mu_i}$. In each experiment, we simulate every distributed computing scheme for 100 times, in each of which the computing time of a node is simulated by using its straggling and shift parameters.

\subsection{Parameter Impact Analysis}
In this sub-section, we investigate the impacts of  parameters in BPCC, including number of batches, number of worker nodes, and the straggling and shift parameters in the computing model. 
\subsubsection{Number of batches}
The  number of batches $p_i$ is an important parameter to configure. In Section \ref{sec:impact_of_p},  
we have theoretically analyzed its impact on the performance of BPCC. Here we conduct simulation studies to demonstrate its impact described in \theoref{thm:convergence}, \theoref{thm:4} and Corollary \ref{corollary:2}. In particular, two experiments are designed. 

In the first experiment, we show that the approximated execution time $\tau^*$ of BPCC given by \equref{eq:tau} decreases with the increase of any $p_i$, $i\in[N]$, as presented in \theoref{thm:convergence}. In particular, we vary the number of batches for one of the worker nodes and fix the number of batches for the others. Specially, we vary $p_1$ and let $p_j=1, \forall j \in [N]\setminus\{1\}$. As shown in \figref{fig:taupa}, $\tau^*$ indeed decreases as $p_1$ increases. 

\begin{figure}[h]
\subfigure[]{
  \label{fig:taupa}
  \includegraphics[width=0.23\textwidth]{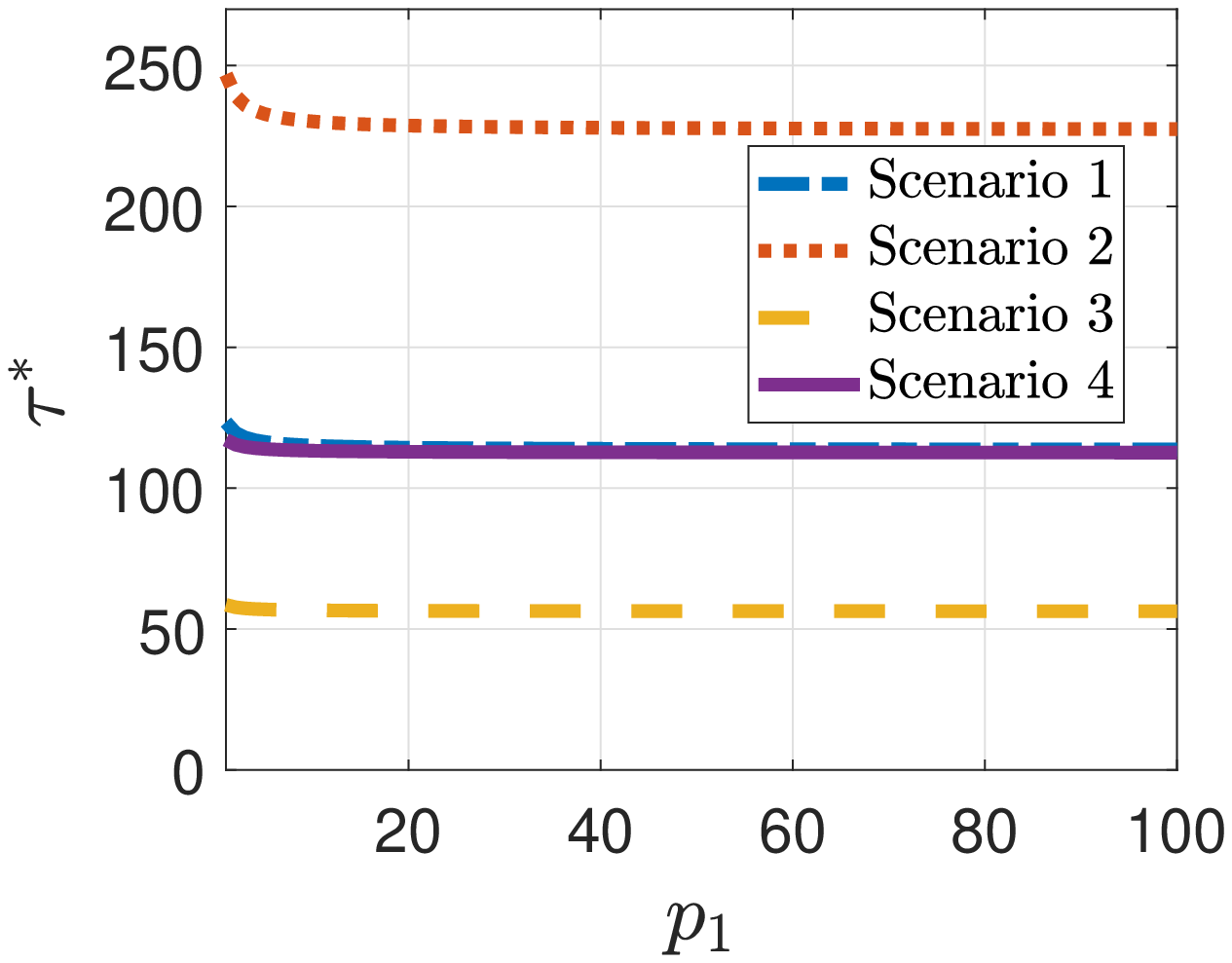}}
\subfigure[]{
\label{fig:taupb}
\includegraphics[width=0.23\textwidth]{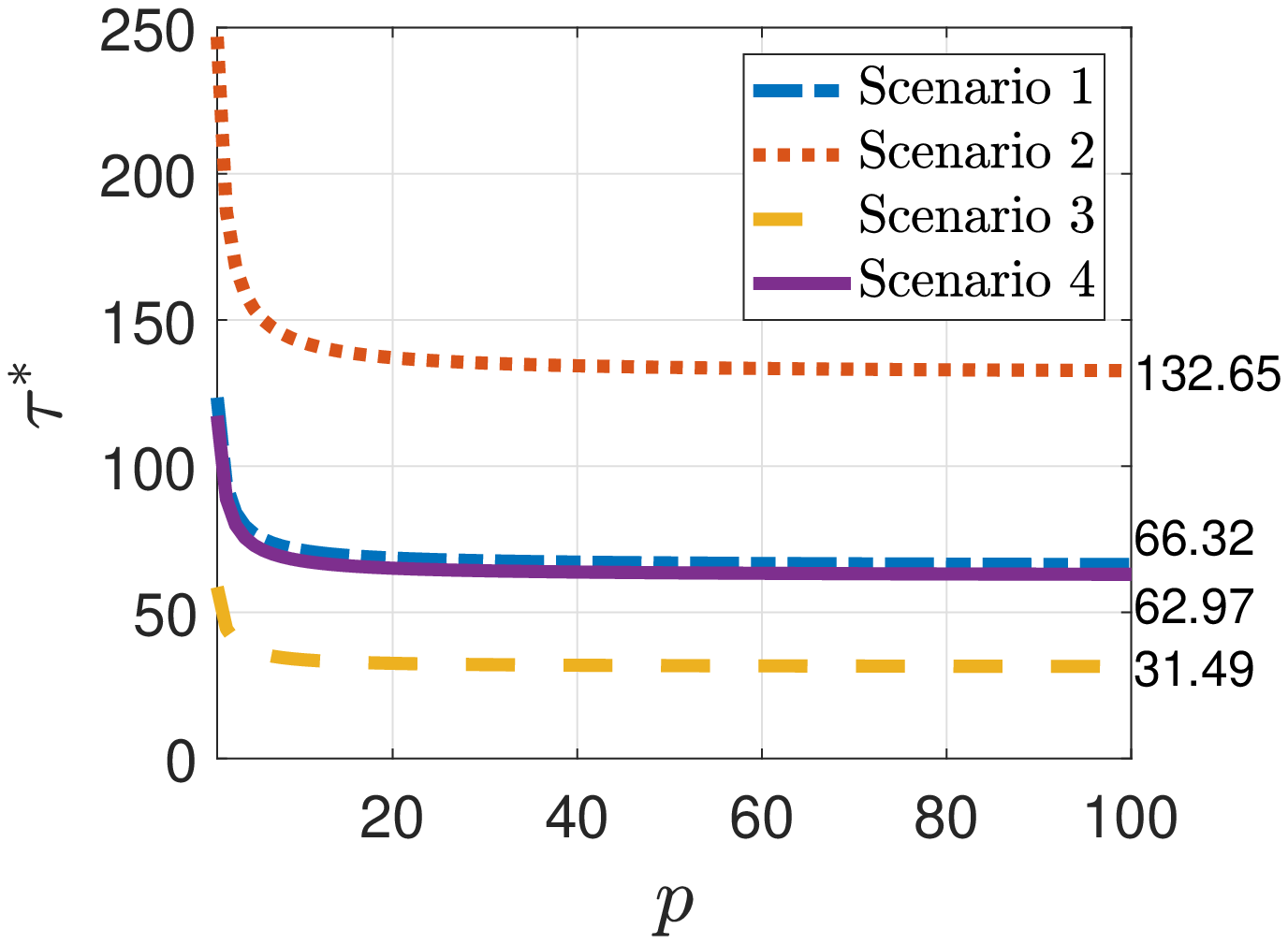}
}
  \caption{The approximated execution time $\tau^*$ of BPCC at different values of   
	a) $p_1$, when $p_j = 1, \forall j \in [N]\setminus\{1\}$, and b) $p$, when $p_i = p, \forall i \in [N]$, in different scenarios.}\label{fig:taup}
\end{figure}

\begin{figure}[h]
\subfigure[]{
\label{fig:l1}
  \includegraphics[width=0.23\textwidth]{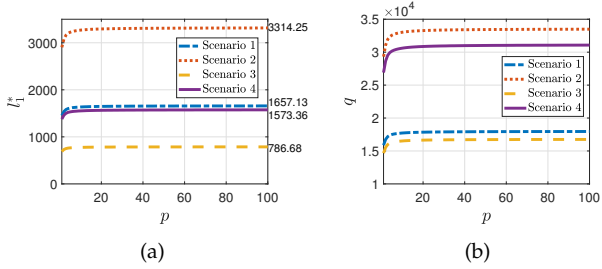}}
\subfigure[]{
\label{fig:suml}
\includegraphics[width=0.23\textwidth]{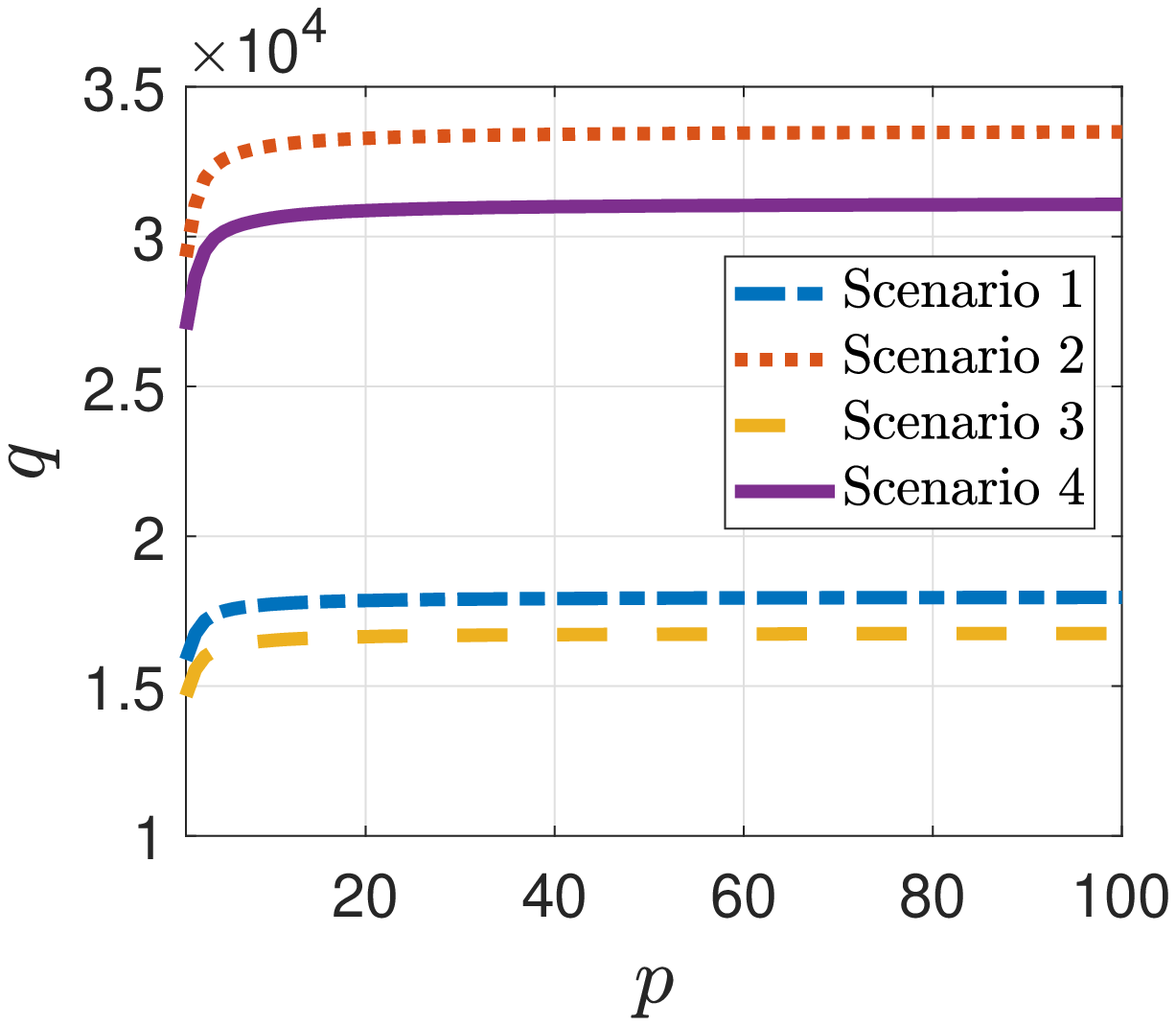}
}
  \caption{The value of a) load $\ell^*_1$ and b) total load $q=\sum_{i=1}^N \ell^*_i$ at different values of $p$, when $p_i = p, \forall i \in [N]$, in different scenarios.}
\end{figure}

In the second experiment, we show that the approximated execution time $\tau^*$ and the load $\ell^*_i$ tend to converge as $p_i$ increases for all $i\in [N]$, as presented in \theoref{thm:4} and Corollary \ref{corollary:2}. In this experiment, we vary $p_i$ simultaneously for all $i \in [N]$. In other words, we let $p_i = p \in \mathbb{Z}^+$, $\forall i \in [N]$ and vary the value of $p$. As shown in \figref{fig:taupb}, $\tau^*$ decreases with the increase of $p$ and finally converges. Note that when $p = 100$, $\tau^*$ equals to  66.32, 132.65, 31.49, 62.97 for the four scenarios, respectively, which are already very close to its theoretical infimum 65.77, 131.54, 31.22, 62.45, computed by \equref{eq:tau_convergence}. \figref{fig:l1} shows the trajectory of the load allocated to one of the worker nodes, i.e., $\ell^*_1$, which decreases and finally converges as $p$ increases. Note that when $p=100$, $\ell^*_1$ equals to 1657.13, 3314.25, 786.68, 1573.36 for the four scenarios, respectively, which are very close to the values of $\hat{\ell}_1$ given by \equref{eq:l_convergence}, i.e., 1659.79, 3319.59, 787.95, 1575.90. Of interest, if we set $p_i = \lfloor\hat{\ell}_i\rfloor$ given by \equref{eq:l_convergence}, $\forall i \in [N]$, $\tau^*$ equals to 65.81, 131.58, 31.26, 62.47 and $\ell^*_1$ equals to 1659.62, 3319.42, 787.73, 1575.68 for the four scenarios, respectively, which are almost the same as the associated  $\inf \tau^*$ and $\hat{\ell}_1$, respectively.

\figref{fig:l1} shows the impact of parameter $p_i$ on the load $\ell^*_i$ for one of the work nodes. In \figref{fig:suml}, we also show its impact on the total load $q = \sum_{i=1}^{N}\ell^*_i$. As we can see, the total load $q$ also 
increases with the increase of $p$, where $p_i = p$, $\forall i\in[N]$. This indicates that a larger $p_i$ will require more storage space at the worker nodes. Note that the worker nodes will stop execution once the master node receives sufficient amount of results for recovering the final result. Therefore, a larger total load $q$ does not increase the computation load for the worker nodes. This study tells us that the configuration of parameter $p_i$ should trade off between computational efficiency and storage consumption.

\begin{figure}[h]
\subfigure[]{
  \includegraphics[width=0.23\textwidth]{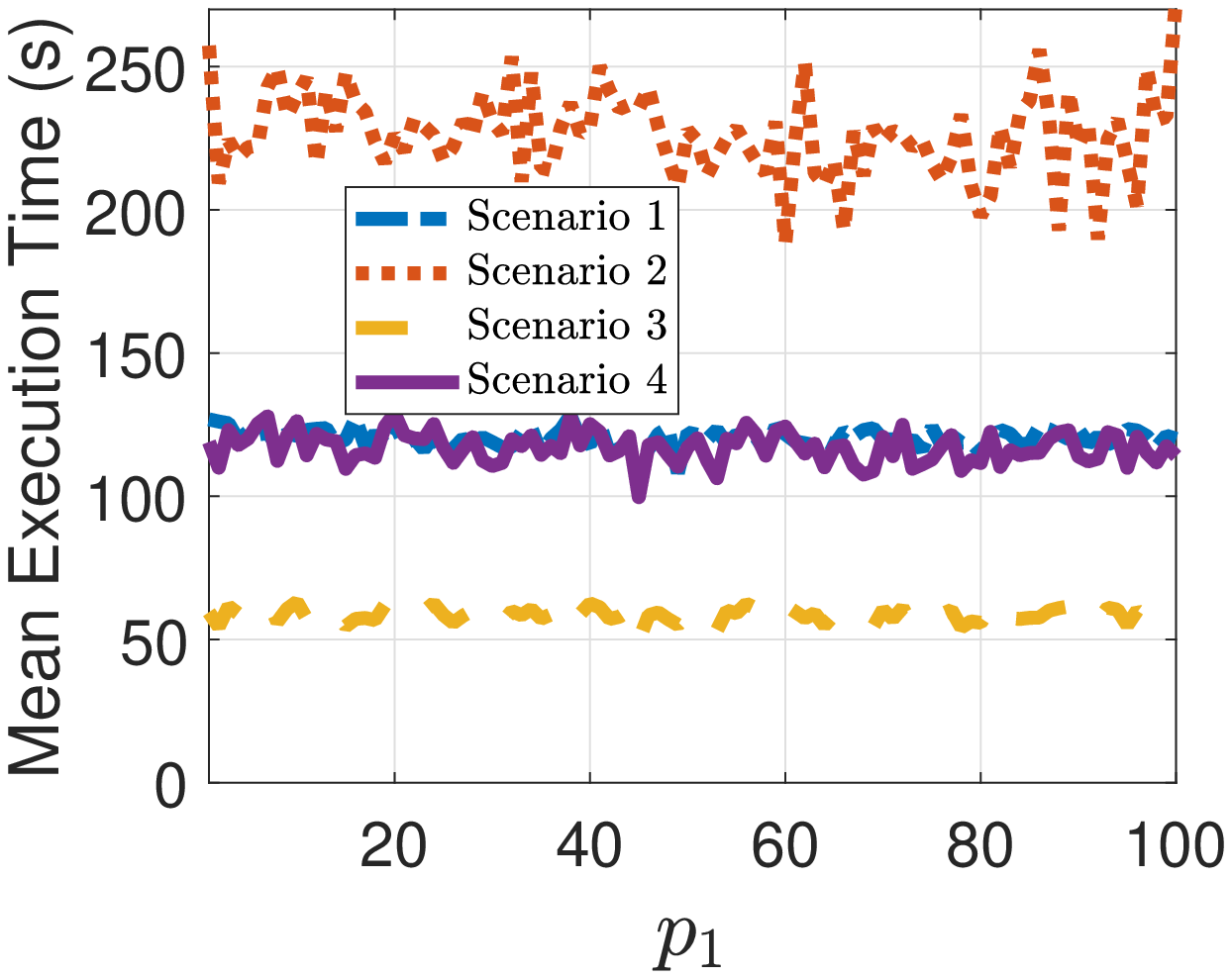}}
\subfigure[]{
\includegraphics[width=0.23\textwidth]{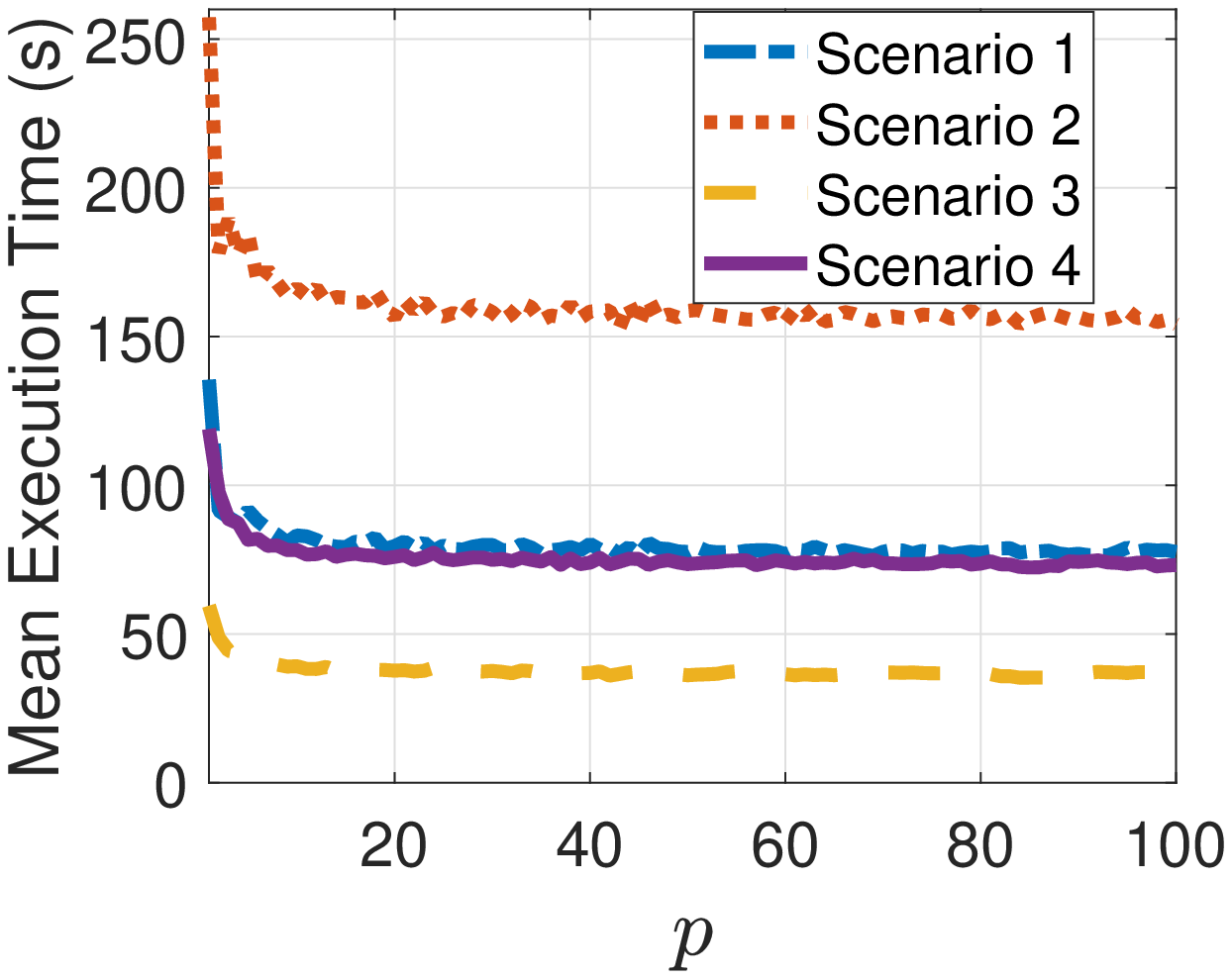}
}
  \caption{The value of a) load $\ell^*_1$ and b) total load $q=\sum_{i=1}^N \ell^*_i$ at different values of $p$, when $p_i = p, \forall i \in [N]$, in different scenarios.}\label{fig:etwithp}
\end{figure}

As $\tau^*$ is an approximation of BPCC's execution time, we also show in \figref{fig:etwithp} the impact of $p_i$ on the expected execution time $\mathbb{E}[T_{\mathrm{BPCC}}]$ of BPCC, which is estimated using the Monte Carlo simulation method, specifically, by repeating each experiment for 100 times and averaging the times to execute the BPCC scheme. Comparing \figref{fig:taup} and \figref{fig:etwithp}, we can see that $\tau^*$ approximates $\mathbb{E}[T_{\mathrm{BPCC}}]$ generally well. The fluctuations are caused by the uncertainty of the computation times and the relatively weak estimation capability of the Monte Carlo method, which requires large number of simulations to obtain an accurate mean estimate. 
As we will show in the next study, the approximation accuracy of $\tau^*$ is impacted by the number of worker nodes $N$.

\subsubsection{Number of worker nodes}
As we have theoretically proved in \theoref{thm:tau}, the approximated execution time $\tau^*$ converges to the true expected execution time $\mathbb{E}[T_{\mathrm{BPCC}}]$ of BPCC, when the number of worker nodes $N$ approaches infinity. To demonstrate this theorem, we vary $N$ and set $r=100N+10000$, and record the approximation error of $\tau^*$, given by $|\tau^*-\mathbb{E}[T_{\mathrm{BPCC}}]|$, for each value of $N$. The results are shown in \figref{fig:N_converge}. Note that, instead of the four scenarios described in \secref{sec:computation_scenario}, we design four new scenarios for this study, where the configuration of each scenario is specified in the figure. As we can see, the approximation error of $\tau^*$ decreases with the increase of the number of worker nodes, and finally converges to zero.  

\begin{figure}[!htb]
\begin{center}
  \includegraphics[width=0.23\textwidth]{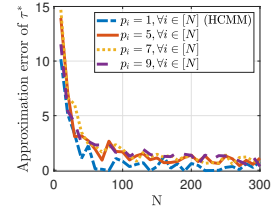}
	\caption{The approximation error of $\tau^*$ at different values of $N$ with $r=100N+10000$.
}\label{fig:N_converge}
\end{center}
\end{figure}

			

\medskip
From the above studies, we can see that, as the number of batches $p_i$ for any worker node $i\in[N]$ increases, the efficiency of BPCC improves, but the demand for storage also increases.  Because storage consumption is not our main concern in this study, in the following experiments, we set $p_i$ to its maximum value possible, i.e.,  $p_i=\lfloor\hat{\ell}_i\rfloor$, $\forall i \in [N]$, considering that a valid $p_i$ should be a positive integer smaller than or equal to $\ell^*_i$ and $\ell^*_i$ converges to $\hat{\ell}_i$ as $p_i$, $\forall i \in[N]$, increases. 

\subsubsection{Straggling and shift parameters}
In BPCC, to determine the load numbers $\ell_i$, 
we need to know the values of the straggling and shift parameters, $\mu_i$ and $\alpha_i$, which are estimated by measuring the actual execution behaviors in real experiments. To understand the impact of parameter estimation errors to the performance of BPCC, we conduct a sensitivity study.  
In particular, to study how sensitive BPCC is to the estimation errors associated with the straggling parameters $\mu_i$, we fix the shift parameters $\alpha_i$ and deviate each $\mu_i$ from its true value by randomly picking a value from the interval $(\mu^{min}_i, \mu^{max}_i)$, where $\mu^{min}_i = \mu^*_i(1-\Delta)$, $\mu^{max}_i = \mu^*_i(1+\Delta)$, $\mu^*_i$ is the true value and $\Delta >0$ represents the degree of deviation. As $\mu_i$ should be positive, we let $\mu^{min}_i = 0$, if $\Delta > 1$. 
\figref{fig:estimation_errora} shows the relative change of the mean execution time, measured by $\frac{\hat{\mathbb{E}}'[T]-\hat{\mathbb{E}}[T]}{\hat{\mathbb{E}}[T]}$, at different values of $\Delta$ in different scenarios, where $\hat{\mathbb{E}}[T]$ and $\hat{\mathbb{E}}'[T]$ are the mean execution time obtained by using the true and erroneous parameter values, respectively. Similarly, we plot in \figref{fig:estimation_errorb} the relative change of the mean execution time when the shift parameters $\alpha_i$ suffer from estimation errors.
As we can see, the deviation of straggling parameters $\mu_i$  has 
less impact on the performance of BPCC than that of shift parameters $\alpha_i$, and BPCC is robust to small errors in general.

\begin{figure}[h]
\begin{center}
\subfigure[]{
\label{fig:estimation_errora}
		\includegraphics[width=0.23\textwidth]{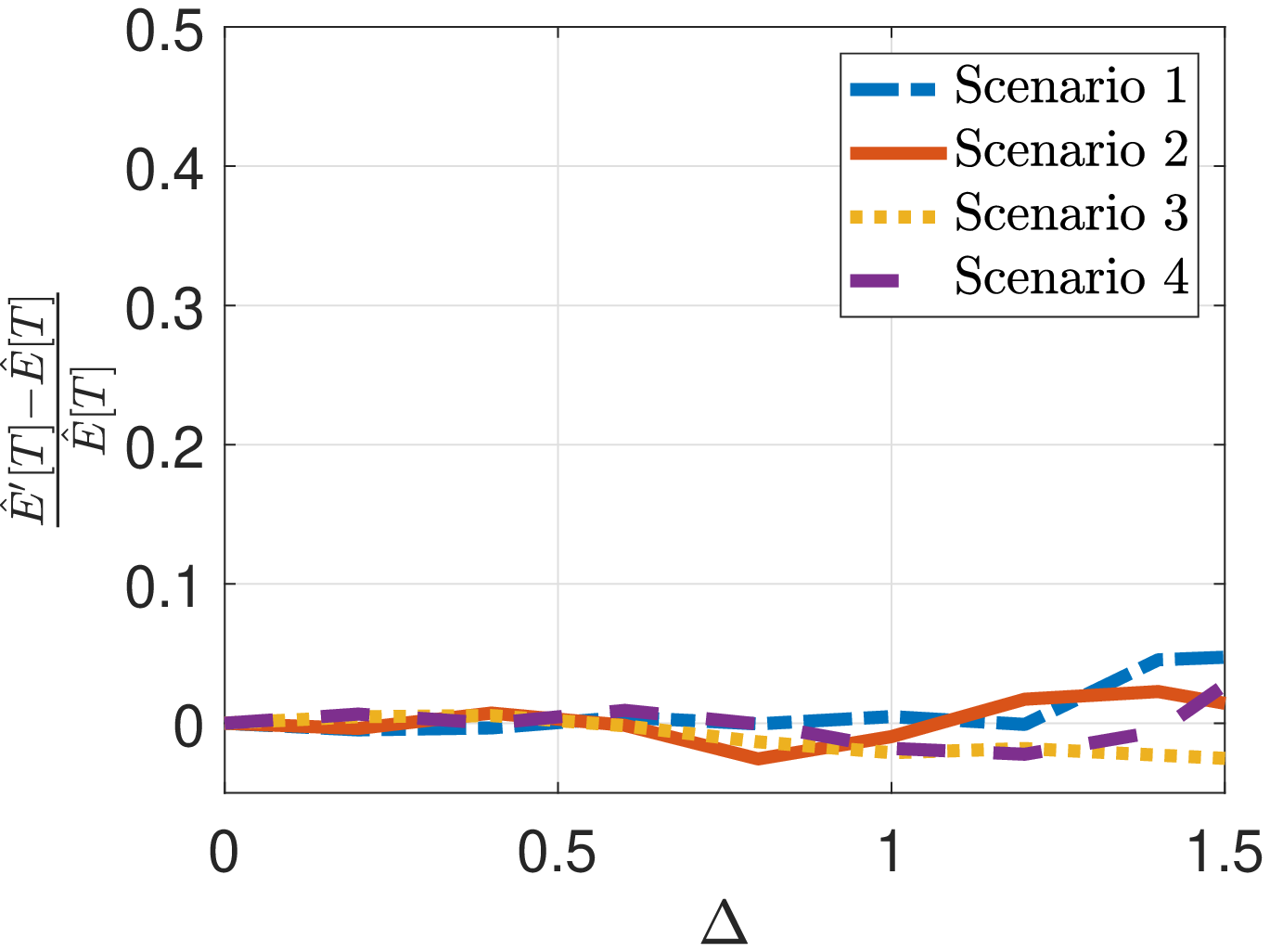}}
\subfigure[]{
\label{fig:estimation_errorb}
        \includegraphics[width=0.23\textwidth]{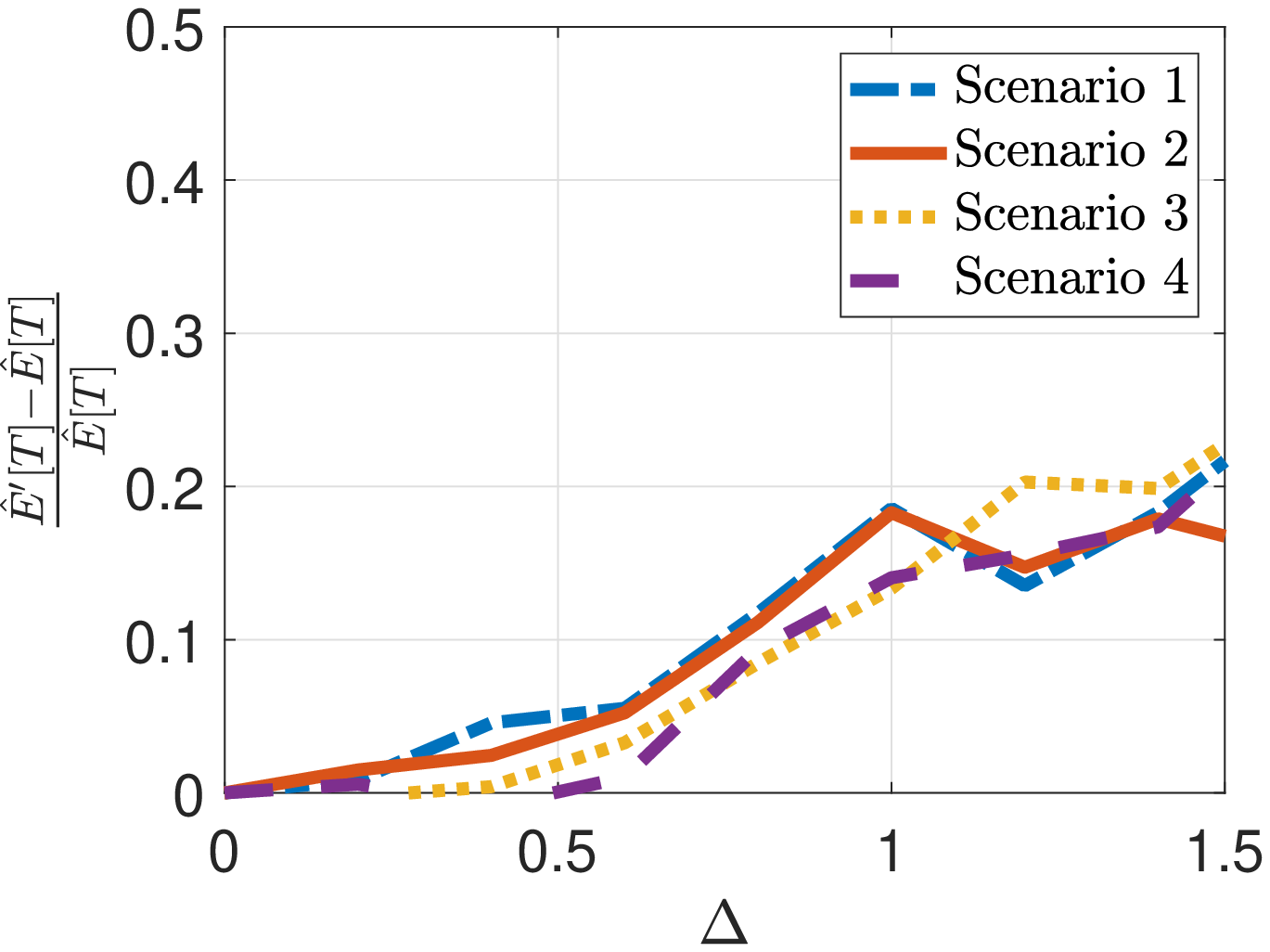}
}
	\caption{Relative change of the mean execution time when a) the straggling parameters $\mu_i$ and b) the shift parameters $\alpha_i$ suffer from different degrees of deviation from their true values in different scenarios.}
\label{fig:estimation_error}
\end{center}
\end{figure}

\subsection{Comparative Performance Studies}
In this sub-section, we compare the performance of the proposed BPCC scheme with three benchmark schemes, including Uniform Uncoded, Load-Balanced Uncoded and HCMM. The parameter $p_i$ in BPCC is set to $p_i= \lfloor\hat{\ell}_i\rfloor$, $\forall i \in [N]$.

\begin{figure}[!htb]
\begin{center}
\subfigure[]{
	\label{fig:fig10221}
\includegraphics[width=0.23\textwidth]{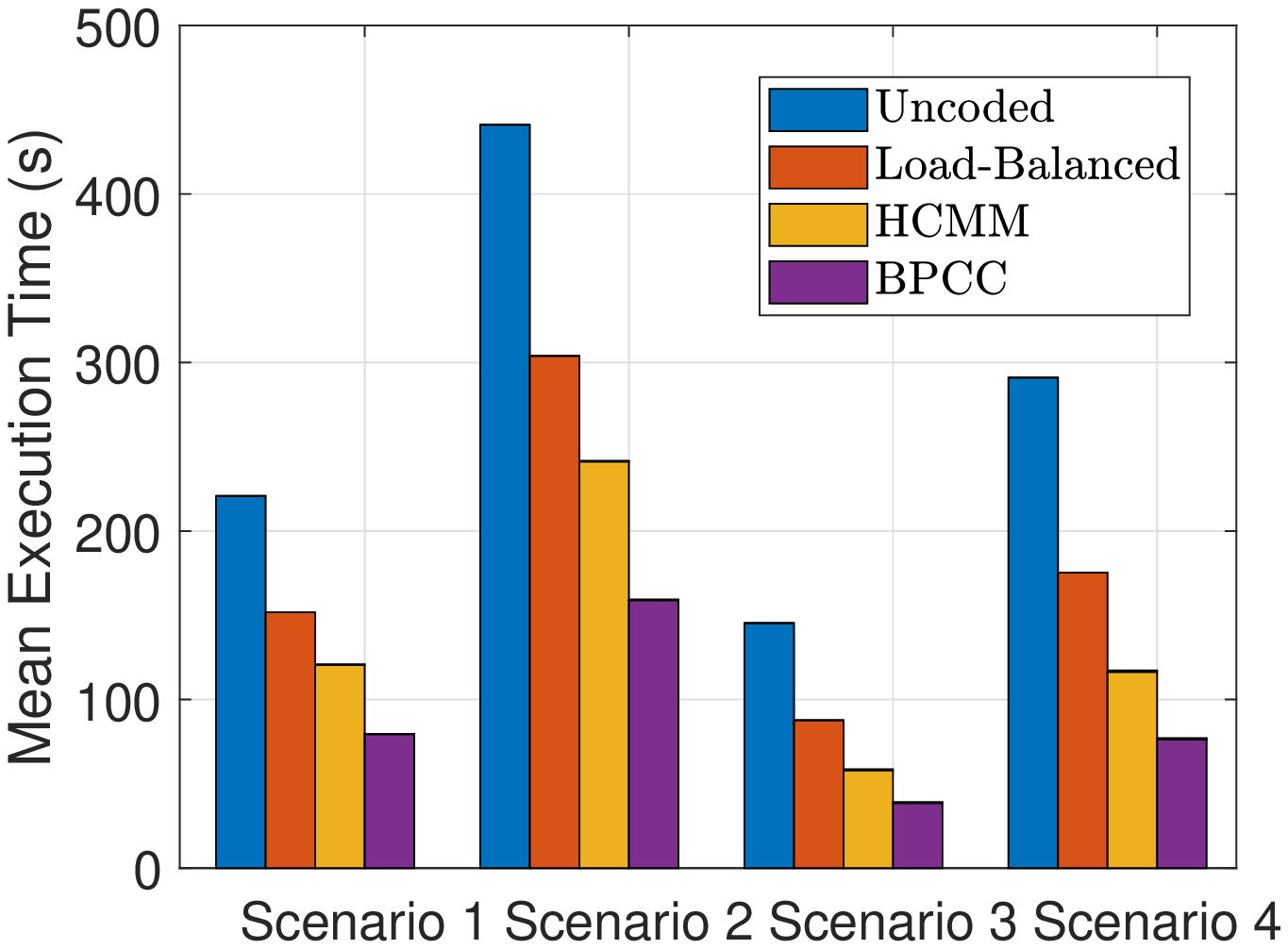}}
\hspace*{-0.6em}
\subfigure[]{
	\label{fig:fig72210}
\includegraphics[width=0.23\textwidth]{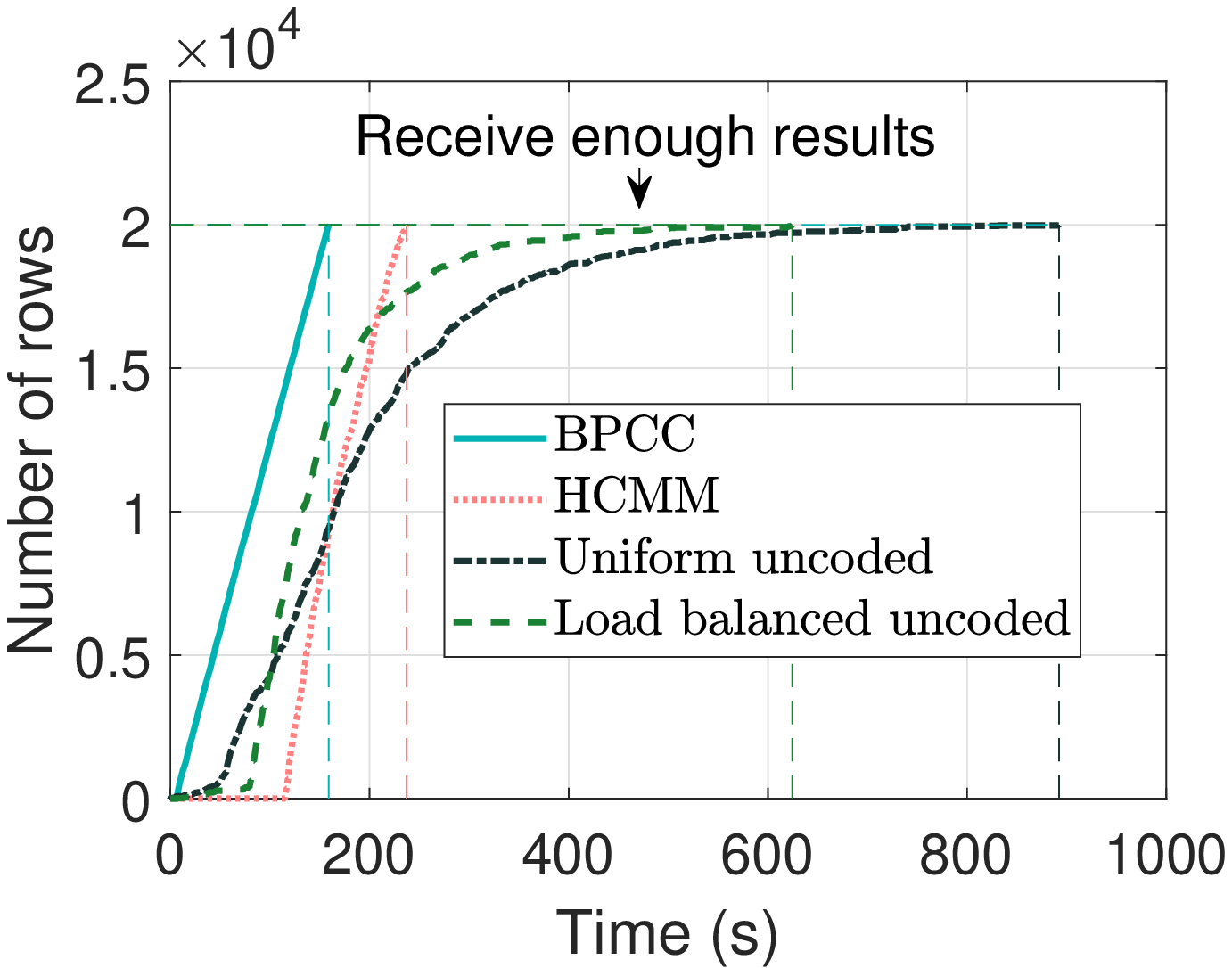}}
  \caption{a) Comparison of the mean execution time of different  schemes in different scenarios. b) The average total number of rows of inner product results received by the master node over time for different  schemes in Scenario 2.
}
\vspace*{-0.6em}

\end{center}
\end{figure}

\figref{fig:fig10221} shows the mean execution times for all schemes, grouped by the computation scenario. We can clearly observe that the proposed BPCC scheme outperforms other benchmark schemes in all scenarios. For instance, BPCC achieves performance improvement of up to $73\%$ over the Uniform Uncoded scheme, up to $56\%$ over Load-Balanced Uncoded scheme, and up to $34\%$ over HCMM. Note that the execution times are directly derived by using the computing model in \equref{eq:shift_exponential} and the decoding times are not considered here.

In \figref{fig:fig10221},  the performance is expressed in terms of the mean execution time, which corresponds to $\mathbb{E}[T_{\mathrm{BPCC}}]$ for different schemes. In \figref{fig:fig72210}, we show the average amount of received results over time for Scenario 2, which corresponds to $\mathbb{E}[S(t)]$ in the theoretical analysis. Remarkably, we can observe from the figure that the master node can quickly receive results from the worker nodes from the very beginning. On the other hand, under the three benchmark schemes, there is a certain duration at the beginning when the master node does not receive any result. This phenomenon occurs because our BPCC scheme allows partial results to be returned, which is very useful for certain applications that can utilize partial results. 
In \figref{fig:fig72210}, we also indicate the time when the master node receives the required amount of results, i.e., $r$. Such a time corresponds to $\tau^*$ (i.e., $\mathbb{E}[S(\tau^*)]=r$).

\section{Experiments on the Amazon EC2 Computing Cluster}
\label{sec:experiment}

In this section, we evaluate the performance of the proposed BPCC scheme in the real distributed computing system. Specifically, we implement three benchmark schemes and the proposed BPCC scheme in the 
Amazon EC2 computing platform \cite{amazonEC2}, which is a classical cloud computing system.

\subsection{Experiment Settings}
To implement the proposed BPCC and the three benchmark schemes over Amazon EC2 clusters \footnote{The source code can be found at  \url{https://github.com/BaoqianWang/Batch-Processing-Coded-Computation}.}, 
we apply a standard distributed computing interface, \emph{Message Passing Interface} (MPI) \cite{gropp99using}, by using an open-source package: \texttt{mpi4py} \cite{mpi4py}, which provides interfaces in Python. 
Moreover, to encode and decode matrices in BPCC, we use the Luby Transform (LT) codes with peeling decoder \cite{mallick2019rateless} that 
 are adopted by HCMM \cite{reisizadeh19coded}. 
The utilization of LT code relaxes the constraint of recovering the final computation result from any $r$ rows to any $r(1+\epsilon)$ rows, where $\epsilon >0$ is desired to be as small as possible. In this study, we adopt the configuration in \cite{reisizadeh19coded} and set $\epsilon=0.13$. The parameter $p_i$ in BPCC is set to $p_i = \lfloor \hat{\ell}_i \rfloor$, $\forall i\in [N]$. 

To evaluate the performance of the four computation schemes, we consider the following four scenarios:

\begin{itemize}
\item \textbf{Scenario 1:} $r = 0.5 \times 10^4$ and $N = 5$, where one \textit{r4.2xlarge} instance, two \textit{r4.xlarge} instances and two \textit{t2.large} instances are used as the worker nodes.  

\item \textbf{Scenario 2:} $r = 1\times 10^4$ and $N = 10$, where two \textit{r4.2xlarge instances} instances, four \textit{r4.xlarge} instances, and four \textit{t2.large} instances are used as the worker nodes. 

\item \textbf{Scenario 3:} $r = 1.5\times 10^4$ and $N = 10$, where 
four \textit{r4.2xlarge} instance and six \textit{r4.xlarge} instances are used as the worker nodes.

\item \textbf{Scenario 4:} $r = 2\times 10^4$ and $N = 15$, 
where seven \textit{r4.2xlarge} instance and eight \textit{r4.xlarge} instances are used as the worker nodes.
\end{itemize}
In all above scenarios, the master node runs in a \textit{m4.xlarge} instance, and the size of the input vector $\boldsymbol{x}\in \mathbb{R}^m$ is set to $m=5\times 10^5$.

\subsection{Parameter Estimation}

In our previous design and analysis, we have assumed that the task completion time $T$ on each node follows a shifted exponential distribution in a general form:
\begin{equation}
\operatorname{Pr}[T \leq t] = 1-e^{-\mu(\frac{t}{r}-\alpha)} = 1-e^{-\frac{\mu}{r}(t - \alpha r)},
\end{equation}
when $t \geq \alpha r$. Therefore, $\mathbb{E}[T] = \frac{r}{\mu} + \alpha r$.

Based on the assumption above, we conduct extensive experiments and measure the actual execution behaviors to estimate the values of the straggling and shift parameters, $\mu$ and $\alpha$, for different types of instances. 
Particularly, let $t_c(r) = \frac{r}{\mu}$ and $t_0(r) = \alpha r$, we run tasks of different sizes. 
For each task size $r$, we execute the task repeatedly for $M=1000$ times and obtain the execution times $\{T_{1}, T_{2}, \ldots, T_{M}\}$. The maximum likelihood estimates of $t_0(r)$ and $t_c(r)$  are then given by $\hat{t}_0(r) = \min_{l\in [M]} T_l$ and $\hat{t}_c(r) = \frac{1}{M}\sum_{l=1}^M T_l - \hat{t}_0(r)$, respectively \cite{pan2002maximum,sharma1977estimation}.
With $\hat{t}_0(r)$ and $\hat{t}_c(r)$ for different task sizes $r$, we 
can then estimate the values of $\mu$ and $\alpha$ by using the least squares estimation.
\figref{fig:cum_time} shows the estimated \emph{cumulative distribution function} (CDF) of the processing time of a t2.xlarge instance when task size $r=500$. The estimated $\alpha$ and $\mu$ 
for different types of Amazon EC2 instances are summarized in \tabref{tab:tab008}. These estimated parameters will be used to allocate computation loads for all computing schemes, except the uniform uncoded scheme.
\begin{figure}[!h]
	\begin{center}
		{

\includegraphics[width=0.23\textwidth]{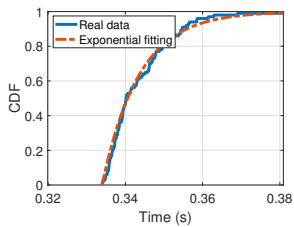}
		
		}
	\end{center}
	\vspace*{-1.0em}
	\caption{ The CDF of the processing time of an Amazon EC2 t2.xlarge instance for computing a task with $r=500$.
}		\label{fig:cum_time}
\end{figure}


\begin{table}[!h] 
	\caption{Estimated computing parameters of different types of Amazon EC2 instances} 
	\label{tab:tab008}
	\centering 
	\begin{tabular}{|>{\centering\arraybackslash} p{1.5cm}| >{\centering\arraybackslash} p{1.5cm}| >{\centering\arraybackslash} p{1.5cm}|}
		\hline 
		\bf Instances & $\mu$ &$\alpha$\\ 
          \hline
         \bf r4.xlarge & 9.42$\times$ $10^4$& 1.75$\times$ $10^{-4}$\\
          \hline
          \bf r4.2xlarge & 9.25$\times$ $10^4$& 1.60$\times$ $10^{-4}$\\
          \hline
        \bf t2.medium  & 2.15$\times$ $10^4$& 5.18$\times$ $10^{-4}$\\
          \hline
        \bf t2.large  & 3.90$\times$ $10^4$& 2.25$\times$ $10^{-4}$\\
          \hline
	\end{tabular}
\end{table}

\subsection{Experimental Results}
To evaluate the performance of the proposed BPCC scheme running on the heterogeneous Amazon clusters, we design three experiments.

\subsubsection{Experiment 1} 
In this experiment, we compare the performance of BPCC with the three benchmark schemes in different scenarios.  For each scenario, we 
run each scheme 100 times and record the mean execution time ($\mathbb{E}[T]$). To evaluate the robustness of these schemes to uncertain disturbances, we introduce unexpected stragglers that are randomly chosen in each run. 
In particular, we randomly select 20\% of the worker nodes to be  stragglers in each run. As stragglers can be slow in computing or returning results (e.g, when communication congestion happens), such stragglers are emulated by delaying the return of computing results such that the computing time observed by the master node is three times of the actual computing time. 
\figref{fig:time_straggler} illustrates the mean execution time of different distributed computing schemes in different scenarios, which also highlights the decoding time required by the coded schemes including BPCC and HCMM. We can see from the figure that the proposed BPCC scheme  outperforms all benchmark schemes in all scenarios. Specifically, the performance improves up to 79\% compared with Uncoded scheme, up to 78\% compared with Load-Balanced scheme, and up to 62\% compared with HCMM. 

\begin{figure}[!h]
\begin{center}
\subfigure[]{
		\includegraphics[width=0.23\textwidth]{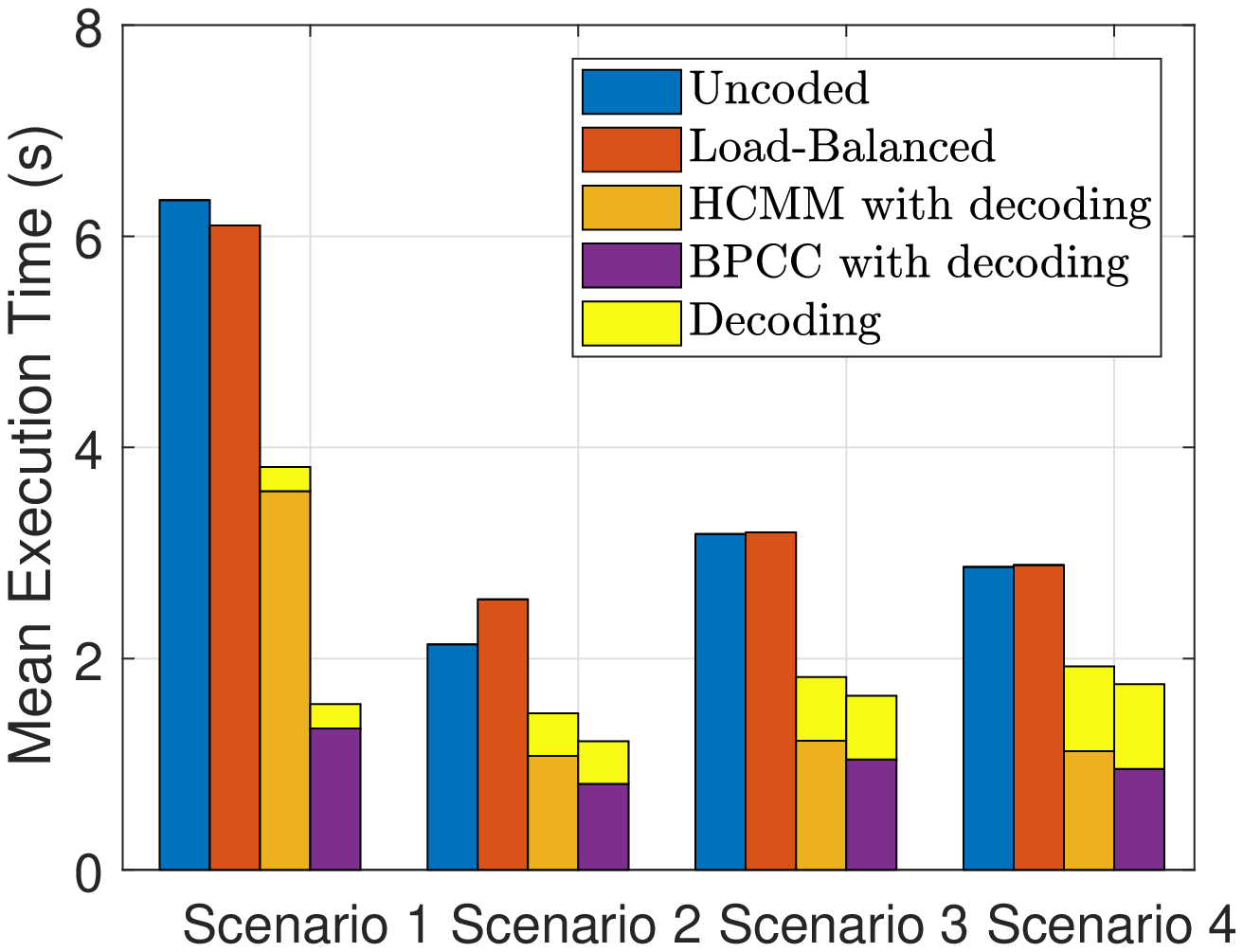}
		\label{fig:time_straggler}}
\subfigure[]{
		\includegraphics[width=0.23\textwidth]{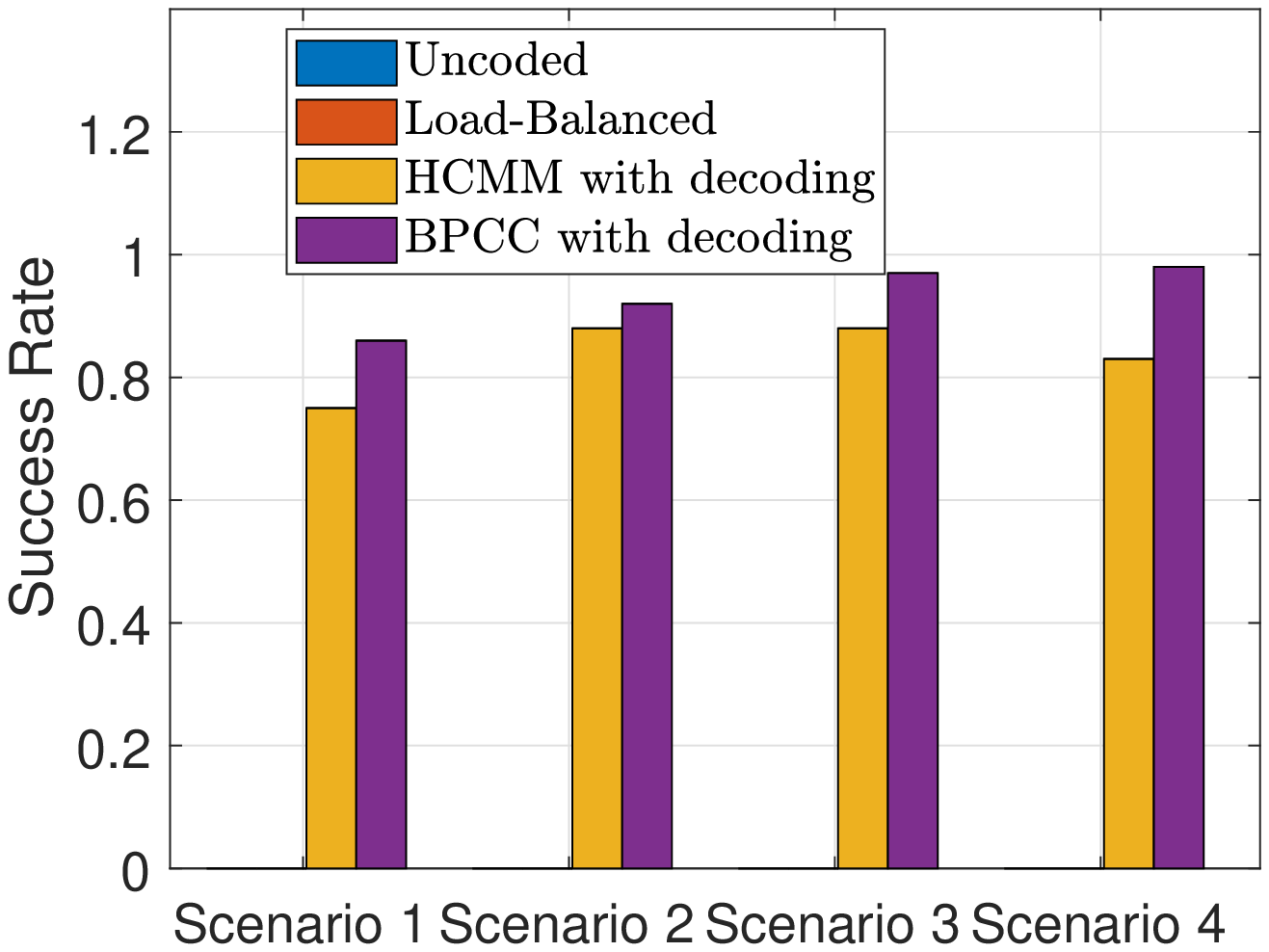}
		\label{fig:success_rate}}
\subfigure[]{
\includegraphics[width=0.23\textwidth]{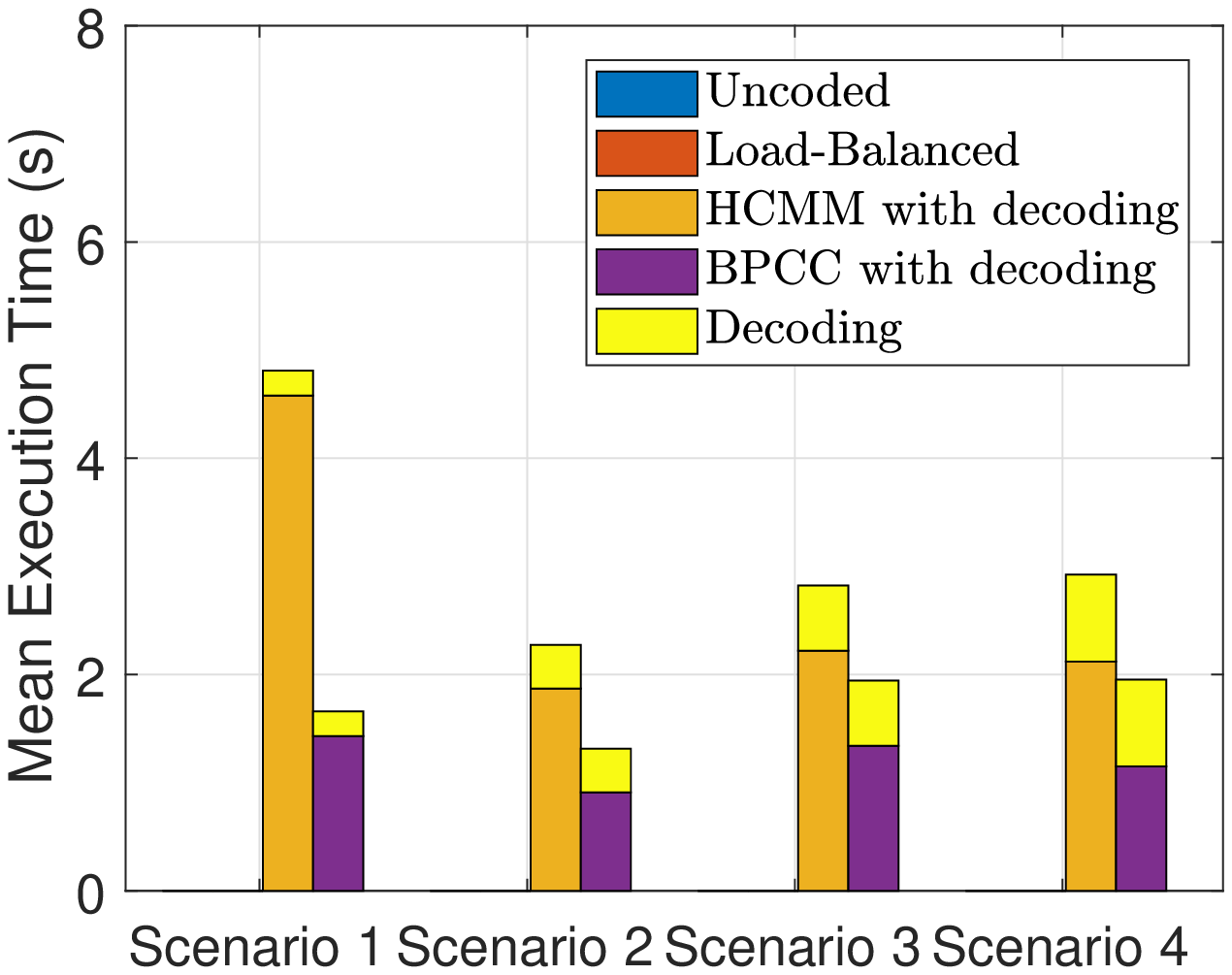}
\label{fig:time_rate}}
\end{center}
	\caption{a) The mean execution time of different schemes in different scenarios at the presence of unexpected stragglers with finite delay. 
	The b) success rate and c) mean execution time of different schemes in different scenarios at the presence of unexpected stragglers with infinite delay.  
	}
\end{figure}

As stragglers can also fail in returning any results (e.g., when nodes/links fail), we also consider such stragglers and emulate them by setting the delay time to infinity. Since no results will be returned by such stragglers, the computation task can fail. \figref{fig:success_rate} shows the success rate (measured by the ratio of successful runs) of each scheme in different scenarios. 
The mean execution time of successful runs is shown  in \figref{fig:time_rate}. As we can see, 
the Uncoded and Load-Balanced schemes fail to complete the task in all runs, as no redundancy is introduced in these schemes. 
Both HCMM and BPCC can  successfully complete the task in most runs, but HCMM has a lower success rate and is less efficient than BPCC.

In \figref{fig:amazon_rows}, we selectively show the average amount of received results over time ($\mathbb{E}[S(t)]$) 
in Scenario 4. As expected, in our BPCC scheme, the master node continuously receives results from the very beginning. However, in other schemes, the master node needs to wait a long time before receiving any result.

\begin{figure}[!h]
\subfigure[]{
\includegraphics[width=0.22\textwidth]{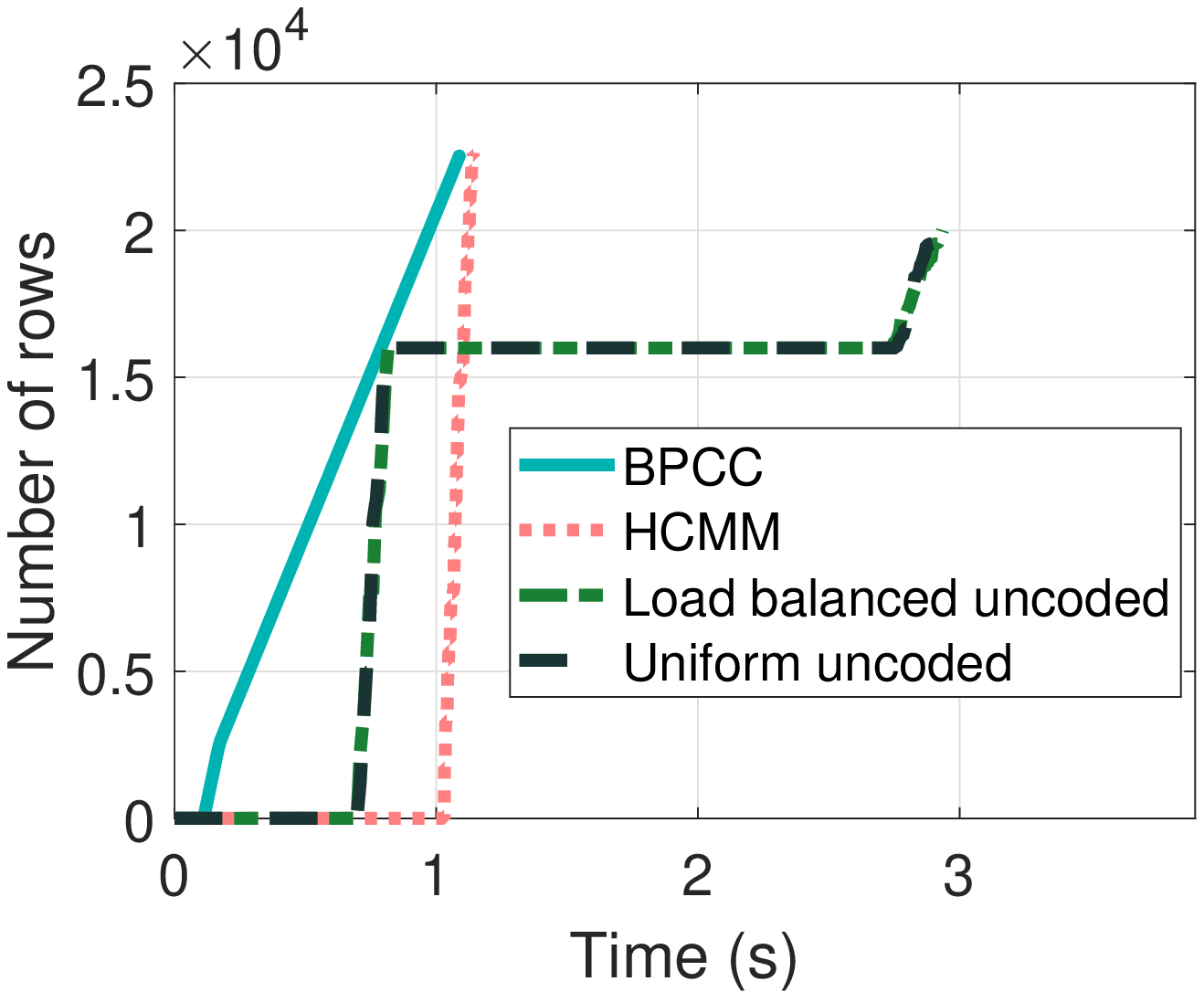}}
\subfigure[]{
\includegraphics[width=0.22\textwidth]{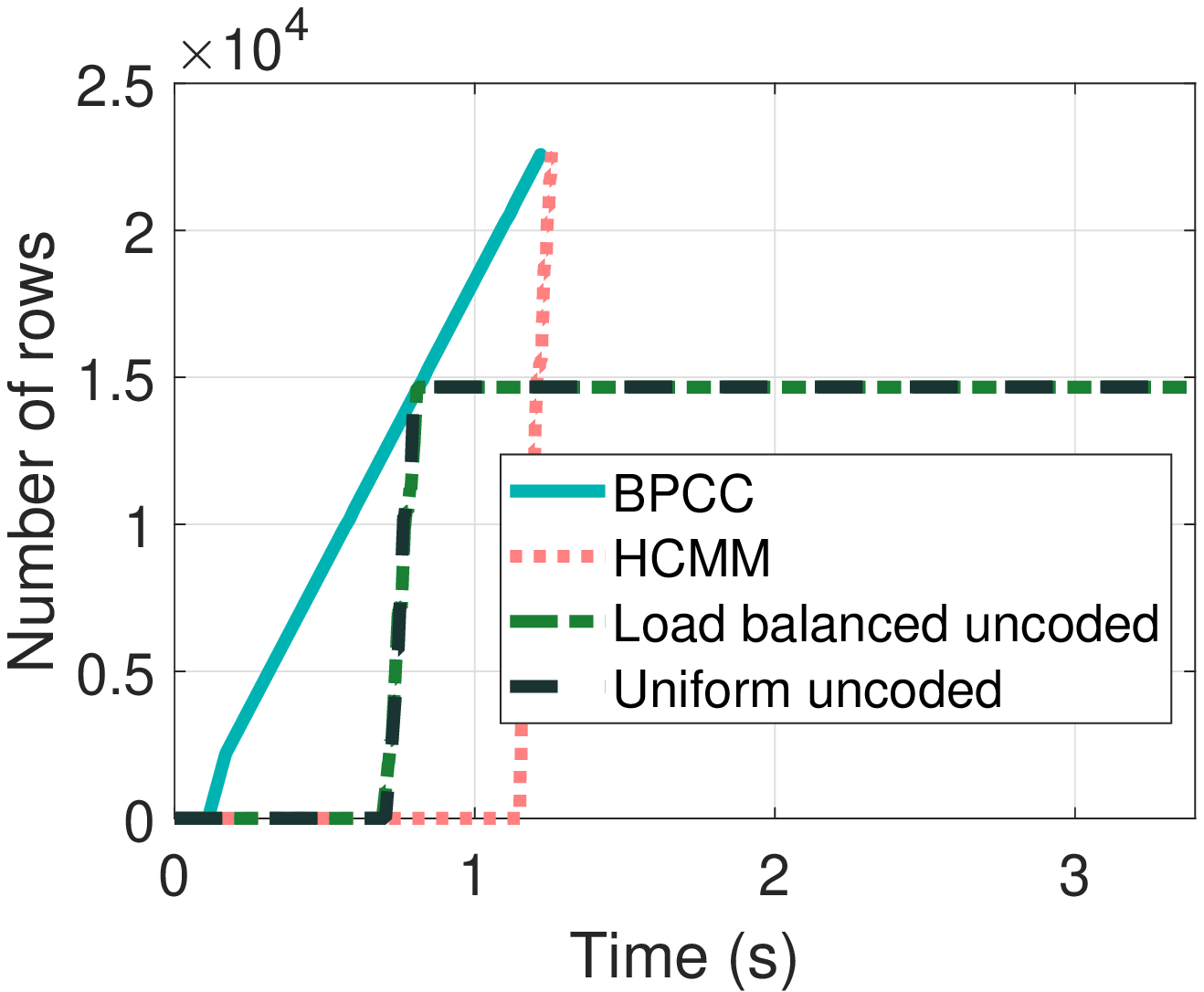}}
  \caption{The average total number of rows of inner product results received by the master node over time for different schemes in Scenario 4 at the presence of unexpected stragglers with a) finite and b) infinite delay.}
  \label{fig:amazon_rows}
\end{figure}

\subsubsection{Experiment 2}

In the second experiment, we study the impact of the number of unexpected stragglers on the performance of the four computation schemes, by varying the percentage of stragglers from 0\% to 60\%. 
Similar as Experiment 1, stragglers can delay in returning results for finite or infinite amount of time.
The mean execution time of each scheme at the presence of different numbers of stragglers with finite delay for Scenario 4 is shown in \figref{fig:time_straggler_prob}. As we can see, when there is no straggler, the Uniform Uncoded scheme and the Load-Balanced Uncoded scheme achieve the best performance, as they do not involve any computation redundancy, compared with the coded schemes.
However, when stragglers exist, our BPCC scheme achieves the best performance, indicating its high robustness to uncertain stragglers. We can also observe from \figref{fig:time_straggler_prob} that
the performances of all schemes degrade with the increase of the number of stragglers. 
Of interest, the performance degradation of the three benchmark schemes slows down when the number of stragglers reaches to a certain value. This is because worker nodes in these schemes won't return any result to the master node until the whole assigned task is completed and all stragglers would delay returning the result for a period that is three times of the task computation time. We also note that the performance of HCMM is even worse than the two uncoded schemes when the percentage of stragglers exceeds $20\%$. This is because each worker node in HCMM is assigned with more computation load, compared with the uncoded schemes, which causes the stragglers in HCMM to wait for a longer time before returning any result. 

In case when stragglers delay in returning results for infinite amount of time, 
the success rate and the mean execution time of each scheme are shown in  \figref{fig:succ_rate_failure} and \figref{fig:time_failure}, respectively. As expected, both the Uncoded and Load-Balanced schemes fail to complete the task. Additionally, the performances of BPCC and HCMM degrade with the increase of the number of stragglers, and BPCC outperforms HCMM. 
\begin{figure}[!h]
\begin{center}
\subfigure[]{
\includegraphics[width=0.23\textwidth]{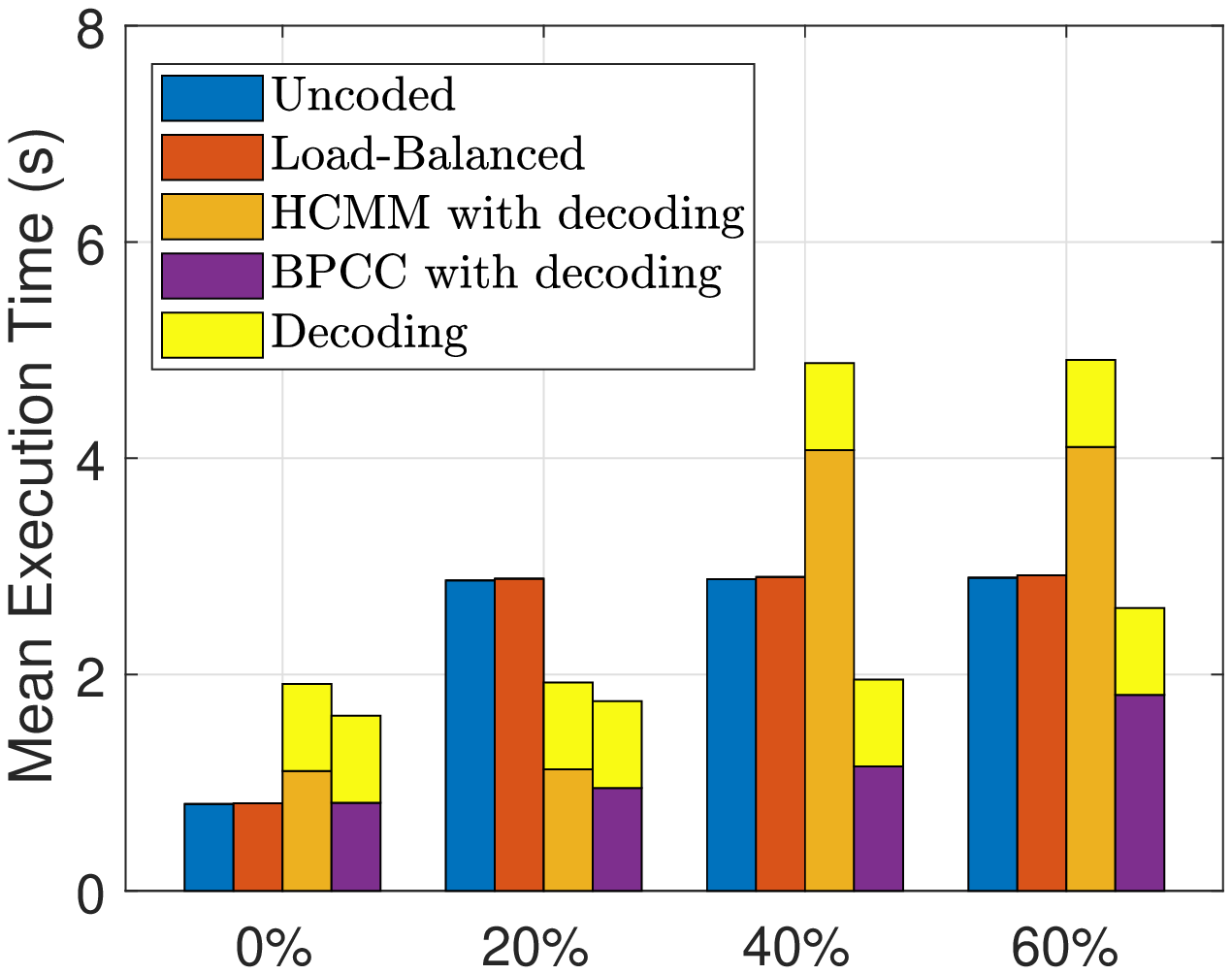}
\label{fig:time_straggler_prob}}
	\subfigure[]{
\includegraphics[width=0.23\textwidth]{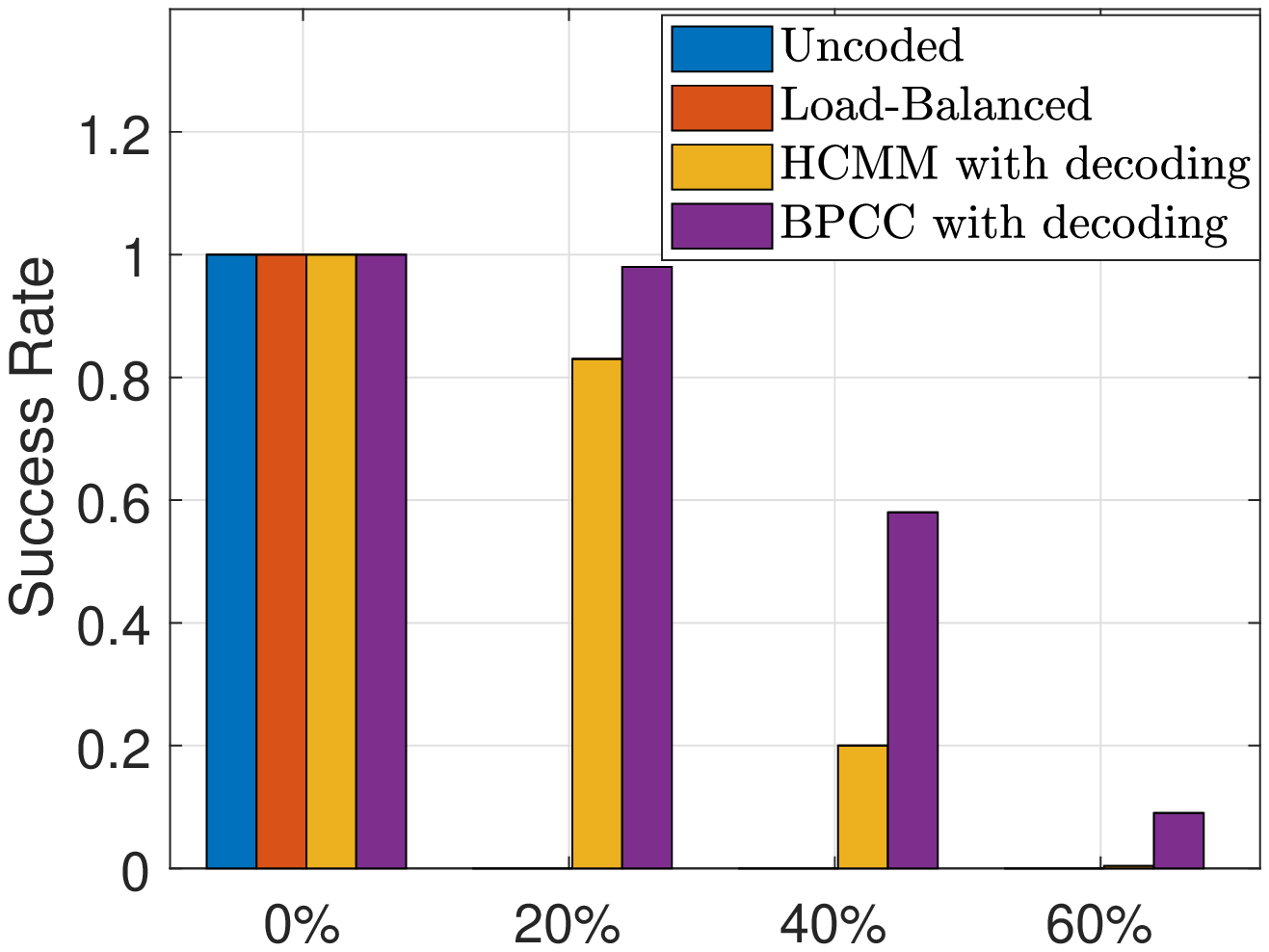}
\label{fig:succ_rate_failure}}
	\subfigure[]{
\includegraphics[width=0.23\textwidth]{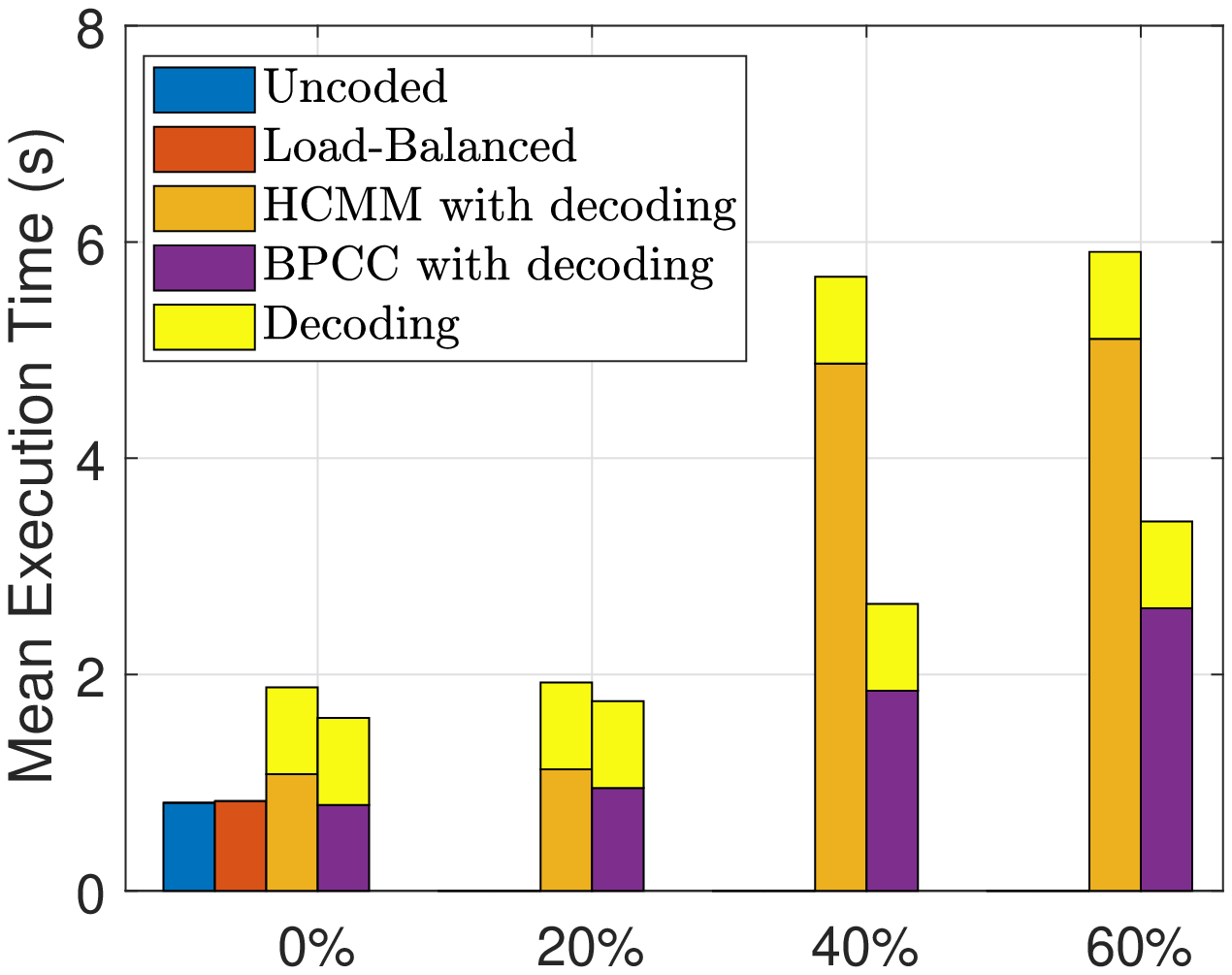}
\label{fig:time_failure}}
	\subfigure[]{
\includegraphics[width=0.23\textwidth]{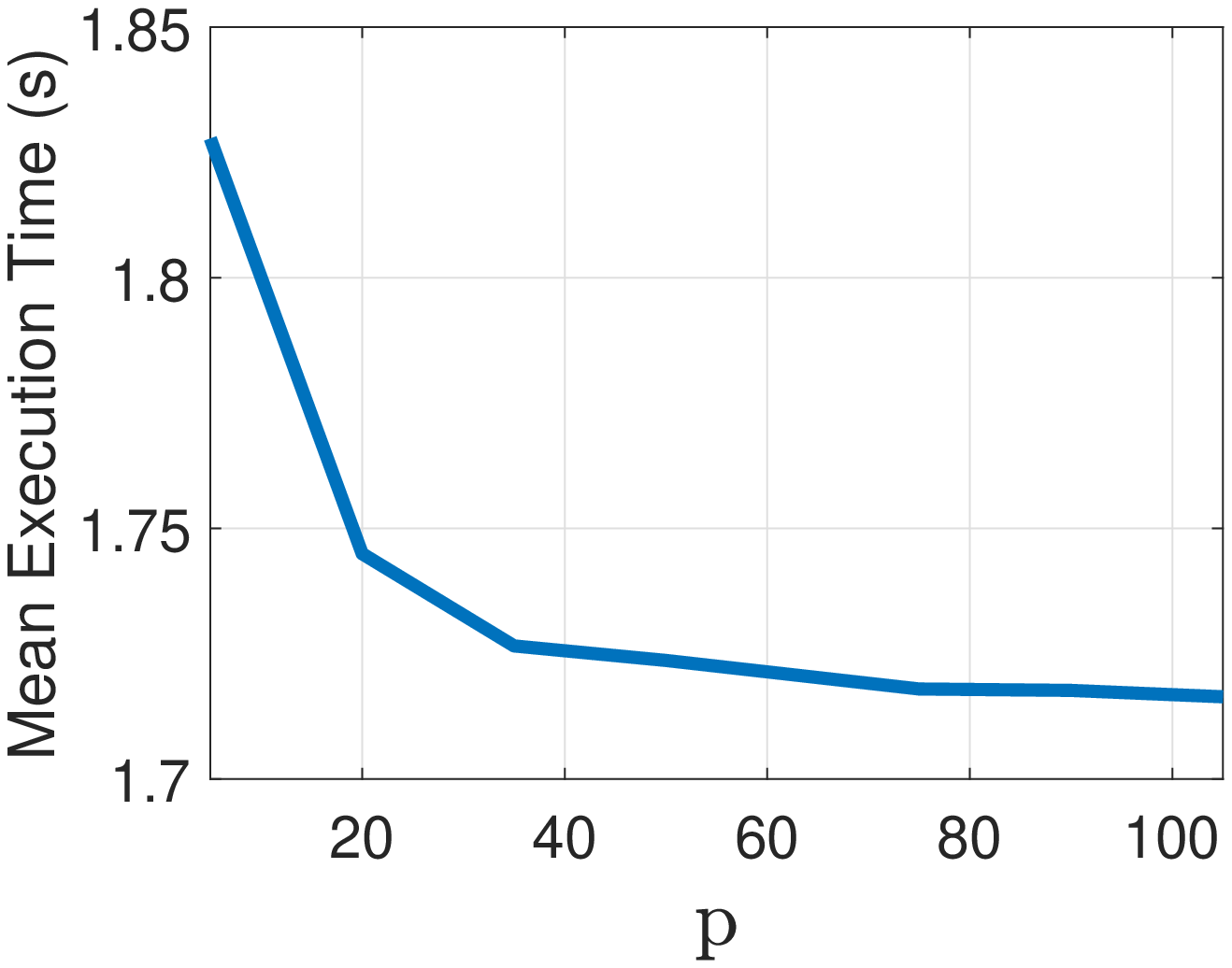}
\label{fig:time_batch}}
	\caption{a) The mean execution time of different schemes in Scenario 4 when different percentages of unexpected stragglers with finite delay are present.  The b) success rate and c) mean execution time of different  schemes in Scenario 4 when different percentages of unexpected stragglers with infinite delay are present. d) The mean execution time of  BPCC at different values of $p$ in Scenario 4.
}
\end{center}
\end{figure}

\subsubsection{Experiment 3}
In the third experiment, we evaluate the impact of the number of batches $p_i$ on  the performance of BPCC running on the Amazon clusters. Similar as the simulation study, we let $p_i = p$, $\forall i \in [N]$, and  vary the value of $p$ from 5 to 100. \figref{fig:time_batch} shows the mean execution time of BPCC  at different values of $p$ under the settings described in Scenario 4 and Experiment 1 when unexpected stragglers with finite delay are present. As expected, the efficiency of BPCC improves with the increase of $p$.

\section{Conclusion}
\label{sec:conclusion}

In this paper, we systematically investigated the design and evaluation of a novel \emph{coded distributed computing} (CDC) framework, namely, \emph{batch-processing based coded computing} (BPCC), for heterogeneous computing systems. The key idea of BPCC is to optimally exploit partial coded results calculated by all distributed computing nodes. Under this BPCC framework, we then investigated a classical CDC problem, matrix-vector multiplication, and formulated an optimization problem for BPCC to minimize the expected task completion time, by configuring the computation load. The BPCC was proved to provide an asymptotically optimal solution and outperform a state-of-the-art CDC scheme for heterogeneous clusters, namely, \emph{heterogeneous coded matrix multiplication} (HCMM). Theoretical analysis reveals the impact of BPCC's key parameter, i.e., number of batches, on its performance, the results of which infer 
the worst and best performance that BPCC can achieve. To evaluate the performance of the proposed BPCC scheme and better understand the impacts of its parameters, we conducted extensive simulation studies and real experiments on 
the Amazon EC2 computing clusters. The simulation and experimental results verify   theoretical results and also demonstrate that the proposed BPCC scheme outperforms all benchmark schemes in  computing systems with uncertain stragglers, in terms of the task completion time and robustness to stragglers. In the future, we will further enhance BPCC by jointly optimizing load allocation and the number of batches to achieve a tradeoff between computational efficiency and storage consumption, and explore its other properties, such as the convergence rate. We will also consider distributed computing systems with mobile computing nodes and other optimization objectives, such as minimizing the energy consumption.

\section*{Acknowledgement}
We would like to thank the National Science Foundation (NSF) under Grants CI-1953048/1730589/1730675/1730570/1730325, CAREER-2048266 and CAREER-1714519  for the support of this work. 

\bibliographystyle{IEEEtran}
\bibliography{IEEEabrv,main}

\clearpage
\newpage
\setcounter{page}{1}

\section*{Appendix}
\subsection*{Proof of Lemma \ref{lemma:lambda}}

To derive the infimum and supremum of $\lambda_i$, we first prove that they are attained at $p_i \rightarrow \infty$ and $p_i=1$, respectively. 
Define the following auxiliary function 
\begin{equation}
h_i(x,y)=(1+\frac{\mu_ix}{y})e^{-\mu_i(\frac{x}{y}-\alpha_i)}, 
\end{equation}
where $x>0$ and $y \in [0,1]$. Note that $h_i(x,y)$ is a monotonically increasing function with respect to $y$, as  $\frac{\partial h_i(x,y)}{\partial y}=\frac{\mu_ix}{y^3}e^{-\mu_i(\frac{x}{y}-\alpha_i)}>0$. 
Also define $P (z) = \{y_0,y_1,y_2,\ldots,y_{z}\}$ as a \textit{partition} of the range of $y$, i.e.,  $[0,1]$,  with $y_k = \frac{k}{z}$, $z \in \mathbb{Z}^+$, $k\in \{0\}\cup[z]$, 
and let 
\begin{align}
M_{k}(x)&=\underset{y}{\sup} \  \{h_i(x, y) | \  y_{k-1} \leq y \leq y_{k}\}\nonumber\\
    &=h_i(x, y_k) = (1+\frac{\mu_ixz}{k})e^{-\frac{\mu_ixz}{k}+\mu_i\alpha_i}
\end{align}

for each $k\in [z]$. We then let $\Delta y_k=y_{k}-y_{k-1}=\frac{1}{z}$, $k\in [z]$, and define
\begin{align}
U(z, x) &= 
\sum_{k=1}^{z} M_{k}(x) \Delta y_{k}\nonumber\\
&=\frac{1}{z}\sum_{k=1}^{z}(1+\frac{\mu_ixz}{k})e^{-\frac{\mu_ixz}{k}+\mu_i\alpha_i}
\end{align}

According to \textbf{Theorem 6.4} in \cite{rudin1964principles}, if there exists another partition $P(z')$ that satisfies $P(z')\supset P(z)$, then 
\begin{equation}
U\left(z', x\right) < U(z, x)
\end{equation}
We can then derive that 
\begin{equation}
\label{eq:range_of_u}
    U(\infty, x) < U(z, x) \leq U(1, x), 
\end{equation}
as $P(z) \supseteq P(1)$ and $ P(z) \subset P(\infty)$, $\forall z \in \mathbb{Z}^+$. The right equality holds when $z = 1$. Furthermore, as $\frac{\partial U(z, x)}{\partial x}=\frac{1}{z}\sum_{k=1}^{z}-\frac{u_i^2zx}{k^2}e^{-\frac{u_izx}{k}+\mu_i\alpha_i}<0$, $U(z,x)$ is a monotonically decreasing function with respect to $x$.

Now let $x = \lambda_i$ and $z = p_i$.  We then have 
\begin{equation}
\label{eq:u_p_lambda}
U(p_i, \lambda_i) = \frac{1}{p_i}\sum_{k=1}^{p_i}(1+\frac{\mu_i\lambda_ip_i}{k})e^{-\frac{\mu_i\lambda_ip_i}{k}+\mu_i\alpha_i} = 1,
\end{equation}
according to \equref{eq:lambda}. Using proof of contradiction, we can then derive that the infimum and supremum of $\lambda_i$ are attained when $p_i \rightarrow \infty$ and $p_i = 1$, respectively. Specifically, suppose $\sup \lambda_i = \bar{\lambda}$ is attained at $p_i = \bar{p} > 1$ and $\lambda^*$ is attained at $p_i = 1$, we then have $\bar{\lambda} \geq \lambda^*$ and $U(1,\bar{\lambda}) \leq U(1, \lambda^*)$ as $U(z,x)$ is a monotonically decreasing function with respect to $x$. Since $U(\bar{p}, \bar{\lambda}) = U(1, \lambda^*)$ according to \equref{eq:u_p_lambda},  we have $U(1,\bar{\lambda}) \leq U(\bar{p},\bar{\lambda})$, which contradicts with $U(1, \bar{\lambda}) > U(\bar{p},\bar{\lambda})$ according to \equref{eq:range_of_u}. Therefore, $\sup \lambda_i$ must be attained at $p_i=1$. Similarly, we can prove that the infimum of $\lambda_i$ is attained when $p_i\rightarrow \infty$.

Next, we find the specific formulas for the infimum and supremum of $\lambda_i$. 
In particular, the infimum of $\lambda_i$ is attained when $p_i \rightarrow \infty$, which can be found by solving \equref{eq:u_p_lambda}. Specifically, 
\begin{equation}
\label{eq:intqinfty}
\lim_{p_i\rightarrow \infty}\frac{1}{p_i}\sum_{k=1}^{p_i}(1+\frac{\mu_i\lambda_ip_i}{k})e^{-\frac{\mu_i\lambda_ip_i}{k}+\mu_i\alpha_i} = 1
\end{equation}
which is equivalent to solving
\begin{equation}
\label{eq:lmm1a}
\int_{0}^{1}(1+\frac{\mu_i\lambda_i}{x})e^{\frac{-\mu_i\lambda_i}{x}}dx=e^{-\mu_i\alpha_i}
\end{equation}
Define variable $v = \frac{\mu_i\lambda_i}{x}$, the term in the left side of the above equation can be simplified as

    \begin{align}
    \label{eq:lmm1b}
&\int_{0}^{1}\left(1+\frac{\mu_i\lambda_i}{x}\right) e^{-\frac{\mu_i\lambda_i}{x}} d x \nonumber\\ 
=&-\mu_i\lambda_i \int_{\mu_i\lambda_i}^{\infty}(1+v) e^{-v} d\left(\frac{1}{v}\right) \nonumber\\
=&-\mu_i\lambda_i\left[(\frac{1}{v}+1)e^{-v}|^{\infty}_{\mu_i\lambda_i}-\int_{\mu_i\lambda_i}^{\infty} \frac{1}{v} d ((1+v)e^{-v})\right]\nonumber\\
=&(1+\mu_i\lambda_i) e^{-\mu_i\lambda_i}-\mu_i\lambda_i \int_{\mu_i\lambda_i}^{\infty} e^{-v} d v\nonumber\\
=&e^{-\mu_i\lambda_i}
\end{align}

Combining \equref{eq:lmm1a} and \equref{eq:lmm1b}, we can get $\lambda_i=\alpha_i$, when $p_i \rightarrow \infty$.

Similarly, we can obtain the supremum of $\lambda_i$ by solving \equref{eq:u_p_lambda} and setting $p_i = 1$. Specifically, 
we aim to solve the following equation
\begin{equation}
(1+\mu_i\lambda_i)e^{-\mu_i\lambda_i+\mu_i\alpha_i} = 1
\end{equation}
According to \cite{corless1996lambertw}, the solution to $a^x  = x +b$ is $x = \frac{-b-W(-a^{-b}\ln a)}{\ln a}$, where $W(\cdot)$ is the Lambert W function. We can then get the following solution to the above equation  

\[\lambda_i=\frac{W(-e^{-\alpha_i\mu_i-1})+1}{-\mu_i},\] 
by letting $x = \mu_i\lambda_i-\alpha_i\mu_i$, $a = e$ and $b = \alpha_i\mu_i+1$.

\subsection*{Proof of Lemma \ref{lemma:lm1}}
To prove $\tau^{*}-o(1) < t^{*} \leq \tau^{*}+o(1)$ in Lemma \ref{lemma:lm1},
we will apply the McDiarmid's inequalities \cite{combes15extension} described as follows. 
For a set of independently distributed random variables, $x_1, x_2, \ldots, x_N \in \mathcal{X}$, if a function $f:~\mathcal{X}^N~\rightarrow~\mathbb{R}$ satisfies the Lipschitz condition: 
\begin{equation}
\left|f\left(x_{1}, \ldots, x_{i}, \ldots, x_{N}\right)-f\left(x_{1}, \ldots, x_{i}^{\prime}, \ldots, x_{N}\right)\right| \leq c_{i},\nonumber 
\end{equation} 
for all $x_{1}, \ldots, x_{N}, x_{i}^{\prime} \in \mathcal{X}$, then, for any $\sigma > 0$, 
\begin{eqnarray}
    \operatorname{Pr}\left[\mathbb{E}[f(X)] - f(X) \geq \sigma \right] 
    &\leq& e^ {-\frac{2 \sigma ^{2}}{\sum_{i=1}^{N} c_{i}^{2}}} \label{eq:McD1}\\
    \operatorname{Pr}\left[f(X) - \mathbb{E}[f(X)] \geq \sigma \right] 
    &\leq& e^ {-\frac{2 \sigma ^{2}}{\sum_{i=1}^{N} c_{i}^{2}}}\label{eq:McD2}
    \end{eqnarray}
where $X = (x_1, x_2, \ldots, x_N) \in \mathcal{X}^N$. To apply the above McDiarmid's inequalities in our problem, 
we define $x_i(t)=s_i(t) b_i$, $\forall i \in [N]$, and further define
\begin{equation}
    X(t) = \sum_{i=1}^{N} x_i(t) = \sum_{i=1}^{N} s_i(t) b_i= S(t).\nonumber
\end{equation}
Clearly, under such definitions, we have $c_i(t) = \ell_i(t)$.
To facilitate further discussions, we let $\delta = \Theta\left( \frac{\log N}{\sqrt{N}} \right) = o(1)$, $\epsilon = \delta^2$. We also summarize the asymptotic scales for the parameters: $r = \Theta(N)$, $\lambda_i = \Theta(1)$, $\beta = \Theta(N)$, $\tau^* = \Theta(1)$, $\ell_i^*(\tau^*)=\Theta(1)$. 

Now, let's prove the first inequality in Lemma \ref{lemma:lm1}: 
$\tau^{*}-o(1) < t^{*}$. Define $t = \tau^* - \delta$. According to \equref{eq:tau}, we can derive 
$$
\beta t = \beta \tau^* - \beta \delta = r - \beta \delta.
$$
Applying the second McDiarmid's inequality in \equref{eq:McD2}, we can then derive
\begin{align}
\label{ineq:L1a}
&\operatorname{Pr}\left[S^{*}(t) \geq r + \epsilon\right] \nonumber\\
=& \operatorname{Pr}\left[S^{*}(t) \geq \mathbb{E}[S^{*}(t)] - \mathbb{E}[S^{*}(t)] + r + \epsilon\right] \nonumber\\
=
& \operatorname{Pr}\left[S^{*}(t) \geq \mathbb{E}[S^{*}(t)] - \beta t + r + \epsilon\right] \nonumber\\
=
& \operatorname{Pr}\left[S^{*}(t) - \mathbb{E}[S^{*}(t)] \geq \beta \delta + \epsilon\right] \nonumber\\
\leq
& e^{ -\frac{2(\beta \delta + \epsilon)^2}{\sum_{i=1}^{N} (\ell_i^*(t))^2} }
\end{align}
Using the asymptotic scales of parameters in the right hand side of  \ineqref{ineq:L1a}, we have
\begin{equation}
\operatorname{Pr}\left[S^{*}(t) \geq r + \epsilon\right] \leq \Theta(e^{-\log^2 N}).
\end{equation}
Consequently, we have 
\begin{equation} \label{ineq:L1b}
\operatorname{Pr}\left[S^{*}(t) < r + \epsilon\right] > 1 - \Theta(e^{-\log^2 N}) = \Theta(1).
\end{equation}
\ineqref{ineq:L1b} shows that, if $t^* \leq \tau^* - \delta $, then the probability $\operatorname{Pr}\left[S^{*}(t^*) < r + \epsilon\right]$ is not $o(\frac{1}{N})$, which does not satisfy the constraint in $\mathcal{P}_{\mathrm{alt}}^{(2)}$. Therefore, $t^* > \tau^* - o(1)$.

Next, we prove the second inequality in \lemmaref{lemma:lm1}: $t^{*} \leq \tau^{*}+o(1)$. Define $t' = \tau^* + \delta$. According to \equref{eq:tau}, we can derive 
$$
\beta t' = \beta \tau^* + \beta \delta = r + \beta \delta.
$$
Applying the first McDiarmid's inequality in \equref{eq:McD1}, we can then derive
\begin{align}  
\label{ineq:L1c}
&\operatorname{Pr}\left[S^{*}(t') \leq r - \epsilon\right] \nonumber\\
=
& \operatorname{Pr}\left[S^{*}(t') \leq \mathbb{E}[S^{*}(t')] - \mathbb{E}[S^{*}(t')] + r - \epsilon\right] \nonumber\\
=
& \operatorname{Pr}\left[S^{*}(t') \leq \mathbb{E}[S^{*}(t')] - \beta t' + r - \epsilon\right] \nonumber\\
=
& \operatorname{Pr}\left[\mathbb{E}[S^{*}(t')] - S^{*}(t') \geq \beta \delta + \epsilon\right] \nonumber\\
\leq
& e^{ -\frac{2(\beta \delta + \epsilon)^2}{\sum_{i=1}^{N} (\ell_i^*(t))^2} }
\end{align}

Using the asymptotic scales of parameters in the right hand side of  \ineqref{ineq:L1c}, we have
\begin{equation}\nonumber 
\operatorname{Pr}\left[S^{*}(t') \leq r - \epsilon\right] \leq \Theta(e^{-\log^2 N}).
\end{equation}
\ineqref{ineq:L1c} shows that $t' = \tau^* + \delta$ can satisfy the constraint in $\mathcal{P}_{\mathrm{alt}}^{(2)}$. Since $t^*$ is the minimal time that satisfies the constraint, $t^* \leq t' =\tau^* + \delta$. Therefore, $t^{*} \leq \tau^{*}+o(1)$.

\subsection*{Proof of Theorem \ref{thm:1}}
The asymptotic optimality of BPCC shown in \equref{eq:thm2}  can be proved by showing that \begin{equation}
\label{eq:alternative_proof}
    t^{*}-o(1) {\leq} \mathbb{E}\left[T_{\mathrm{OPT}}\right] {\leq}\mathbb{E}\left[T_{\mathrm{BPCC}}\right] {\leq} t^{*}+o(1) \nonumber
\end{equation}
Since $\mathbb{E}\left[T_{\mathrm{OPT}}\right] \leq \mathbb{E}\left[T_{\mathrm{BPCC}}\right]$ is straightforward because $\mathbb{E}[T_\mathrm{OPT}]$ is the optimal value of $P'_\mathrm{main}$, we  use two steps, inspired by \cite{reisizadeh19coded}, to prove the other two inequalities.

\medskip
\noindent\textbf{Step 1: To prove $t^{*}-o(1) \leq \mathbb{E}\left[T_{\mathrm{OPT}}\right]$.}

Let $\boldsymbol{\ell}_{\mathrm{OPT}}=\left(\ell_{\mathrm{OPT}, 1}, \cdots, \ell_{\mathrm{OPT}, N}\right)$ be the optimal load allocation obtained by solving $\mathcal{P'}_{\mathrm{main}}$ and let $S_{\mathrm{OPT}}(t)$ be the amount of results received by the master node by time $t$ under load allocation $\boldsymbol{\ell}_{\mathrm{OPT}}$. The inequality above can be proved by showing the following inequalities:
\begin{equation}
\label{eq:eq12345}
    t^{*}-\delta_{2}-\delta_{1}\stackrel{(b)}{ \leq} \tau^{*}_{\mathrm{OPT}}-\delta_{1} \stackrel{(a)}{ \leq}  \mathbb{E}\left[T_{\mathrm{OPT}}\right]\nonumber
\end{equation}
where $\tau^{*}_{\mathrm{OPT}}$ is the solution to $\mathbb{E}\left[S_{\mathrm{OPT}}(t)\right]=r$, and $\delta_{1}$ and $\delta_{2}$ are both $\Theta\left(\frac{\log N}{\sqrt{N}}\right)=o(1)$.
To prove Ineq.~($a$), we first define an auxiliary function $g_i(t)$ for each node $i$ as 
\begin{equation}
g_i(t)=1-\frac{1}{p_i}\sum_{k=1}^{p_i}e^{-\mu_i(\frac{t{p_i}}{k\ell_{OPT,i}}-\alpha_i)}.\nonumber
\end{equation}
According to \equref{eq:eqest}, we have 
\begin{equation}
    \mathbb{E}\left[S_{\mathrm{OPT}}(t)\right]
    =\sum_{i=1}^{N} \ell_{\mathrm{OPT}, i}g_i(t)\\ \nonumber
\end{equation}
and
\begin{equation}
\begin{aligned} 
&r-\mathbb{E}[S_{\mathrm{OPT}}\left(\tau^{*}_{\mathrm{OPT}}-\delta_{1}\right)] \\
=~&
\mathbb{E}\left[S_{\mathrm{OPT}}(\tau^{*}_{\mathrm{OPT}})\right]-\mathbb{E}[S_{\mathrm{OPT}}\left(\tau^{*}_{\mathrm{OPT}}-\delta_{1}\right)]\\
=~&
\sum_{i=1}^{N} \ell_{\mathrm{OPT}, i}\left[g_i(\tau^{*}_{\mathrm{OPT}})-g_i(\tau^{*}_{\mathrm{OPT}}-\delta_{1})\right] \\ 
=~&
\sum_{i=1}^{N} \ell_{\mathrm{OPT}, i}\left(\frac{d g_i(\tau^{*}_{\mathrm{OPT}})}{d \tau^{*}_{\mathrm{OPT}}}  \delta_{1}+\mathcal{O}\left(\delta_{1}^{2}\right)\right) \\
\end{aligned}\nonumber
\end{equation}
According to our previous discussions, $r=\Theta(N)$, so $\ell_{OPT,i}=\Theta(1)$, $\forall i \in [N]$. Therefore, $g_i(t)$ does not change with $N$, i.e., $g_i(t)=\Theta(1)$. We then have
\begin{equation}
    \begin{aligned} r-\mathbb{E}\left[S_{\mathrm{OPT}}\left(\tau^{*}_{\mathrm{OPT}}-\delta_{1}\right) \right]  &=\Theta\left(N \delta_{1}\right)+\mathcal{O}\left(N \delta_{1}^{2}\right)\\
    &=\Theta\left(N \delta_{1}\right) \end{aligned}\nonumber
\end{equation}
By using the McDiarmid's inequality in \equref{eq:McD2}, we have 
\begin{equation}
\begin{aligned} 
    &\operatorname{Pr}\left[ S_{\mathrm{OPT}} \left(\tau^{*}_{\mathrm{OPT}}-\delta_{1}\right) \geq r \right]\\ 
    =&\operatorname{Pr}\{S_{\mathrm{OPT}}\left(\tau^{*}_{\mathrm{OPT}}-\delta_{1}\right)-\mathbb{E}\left[S_{\mathrm{OPT}}\left(\tau^{*}_{\mathrm{OPT}}-\delta_{1}\right)\right]\geq\\ 
    & \ \ r-\mathbb{E}\left[S_{\mathrm{OPT}}\left(\tau^{*}_{\mathrm{OPT}}-\delta_{1}\right)\right\}\\
    \leq& ~e^ {-\frac{2\left(\mathbb{E}\left[S_{\mathrm{OPT}}\left(\tau^{*}_{\mathrm{OPT}}-\delta_{1}\right)\right]-r\right)^{2}}{\sum_{i=1}^{N} \ell_{\mathrm{OPT}, i}^{2}}}\\
    =&~e^{{-\Theta\left(N \delta_{1}^{2}\right)}}=o\left(\frac{1}{N}\right),
\end{aligned}\nonumber
\end{equation}
which implies that $\mathbb{E}[T_{\mathrm{OPT}}]\geq \tau^{*}_{\mathrm{OPT}}-\delta_1$.

Next, we proceed to prove Ineq.~($b$). Since $\mathbb{E}[S^{*}(t)]$ is the optimal value of $P_{\mathrm{alt}}^{(1)}$, we have $\mathbb{E}[S^{*}(t)]\geq \mathbb{E}[S_{\mathrm{OPT}}(t)]$.  Moreover, as $\mathbb{E}[S^{*}(\tau^{*})]= r$ according to \equref{eq:alter2}, $\mathbb{E}[S_\mathrm{OPT}(\tau^*_{\mathrm{OPT}})]=r$, and both $\mathbb{E}[S^{*}(t)]$ and $\mathbb{E}[S_{\mathrm{OPT}}(t)]$ increase monotonically with $t$, we can derive
\begin{equation}
    \tau^{*}_{\mathrm{OPT}} \geq \tau^{*} \nonumber
    \label{eq:eq567}
\end{equation}
According to \lemmaref{lemma:lm1}, 
\begin{equation}
   \tau^{*} \geq t^{*}-\delta_{2} \nonumber
\end{equation}
Therefore,
\begin{equation}
\begin{gathered}
   \tau^{*}_{\mathrm{OPT}} - \delta_1 \geq t^{*}-\delta_1-\delta_{2} \nonumber
\end{gathered}
\end{equation}
We have now proved $t^{*}-o(1) \leq \mathbb{E}\left[T_{\mathrm{OPT}}\right]$.

\medskip
\noindent\textbf{Step 2: To prove $\mathbb{E}\left[T_{\mathrm{BPCC}}\right] \leq t^{*}+o(1)$.}

Let $T_{\text{max}}$ be a random variable that denotes the time required for all worker nodes to complete their tasks assigned using BPCC. Let $\mathcal{E}_{1}=\left\{T_{\text{max} }>\Theta(N)\right\}$ and $\mathcal{E}_{2}=\left\{T_{\mathrm{BPCC}}>t^{*}\right\}$ be two events. $\mathbb{E}[T_{\mathrm{BPCC}}]$ can then be computed by

\begin{align} 
\label{eq:eqtwoevents}
\mathbb{E}\left[T_{\mathrm{BPCC}}\right]=
& 
\mathbb{E}\left[T_{\mathrm{BPCC}} | \mathcal{E}_{1}\right]      \operatorname{Pr}\left[\mathcal{E}_{1}\right] \nonumber\\ &
+ 
\mathbb{E}\left[T_{\mathrm{BPCC}} | \mathcal{E}_{1}^{c} \cap \mathcal{E}_{2}\right] \operatorname{Pr}\left[\mathcal{E}_{1}^{c} \cap \mathcal{E}_{2}\right] \nonumber\\ &
+
\mathbb{E}\left[T_{\mathrm{BPCC}} | \mathcal{E}_{1}^{c} \cap \mathcal{E}_{2}^{c}\right] \operatorname{Pr}\left[\mathcal{E}_{1}^{c} \cap \mathcal{E}_{2}^{c}\right] 
\end{align}

The first term in the right hand side of \equref{eq:eqtwoevents} can be written as

\begin{align} 
&\mathbb{E}\left[T_{\mathrm{BPCC}} | \mathcal{E}_{1}\right] \operatorname{Pr}\left[\mathcal{E}_{1}\right]\nonumber\\
=~& 
\mathbb{E}\left[T_{\mathrm{BPCC}} | T_{\max }>\Theta(N)\right] \times \operatorname{Pr}\left[T_{\max }>\Theta(N)\right] \nonumber\\ 
\leq ~&
\mathbb{E}\left[T_{\max } | T_{\max }>\Theta(N)\right] \times \operatorname{Pr}\left[T_{\max }>\Theta(N)\right] \nonumber\\
=~&
\int_{\Theta(N)}^{\infty} t f_{\max }(t) d t,
\end{align}

where $f_\text{{max}}(t)$ is the probability density function (PDF) of $T_{\text{max}}$. A stochastic upper bound of $T_{\text{max}}$ can be found by using $N$ worker nodes that all take the smallest straggling parameter $\min\{\mu_i\}$ and the largest shift parameter $\max\{\alpha_i\}$. 
Using the PDF of the maximum of $N$ i.i.d. exponential random variables, we then have
\begin{align}
\label{eq:thirdterm2}
&\mathbb{E}\left[T_{\mathrm{BPCC}} | \mathcal{E}_{1}\right] \operatorname{Pr}\left[\mathcal{E}_{1}\right] \nonumber\\\leq&
\int_{\Theta(N)}^{\infty} t f_{\max }(t) d t \nonumber\\
\leq &\int_{\Theta(N)}^{\infty} tN k_{1} e^{-k_{1} t}\left(1-e^{-k_{1} t}\right)^{N-1} d t \nonumber\\ 
\leq& \int_{\Theta(N)}^{\infty} N k_{1} t e^{-k_{1} t} d t\nonumber\\
=&-N(t+\frac{1}{k_1})e^{-k_1t} |^{\infty}_{t=\Theta(N)}\nonumber\\
=&~o(1)
\end{align}
where $k_1$ is a constant, i.e., $k_1=\Theta(1)$.

The second term in the right hand side of \equref{eq:eqtwoevents} can be written as
    \begin{align}
    \label{eq:secondTerm}
    &\mathbb{E}[T_{\mathrm{BPCC}} | \mathcal{E}_{1}^{c} \cap \mathcal{E}_{2}]
    \operatorname{Pr}\left[\mathcal{E}_{1}^{c} \cap \mathcal{E}_{2}\right] \nonumber\\
    =&~\mathbb{E}\left[T_{\mathrm{BPCC}} | T_{\max } \leq \Theta(N), T_{\mathrm{BPCC}}>t^{*}\right] \nonumber\\&
    \times \operatorname{Pr}\left[T_{\max } \leq \Theta(N), T_{\mathrm{BPCC}}>t^{*}\right] \nonumber\\ 
    \leq&~ \mathbb{E}\left[T_{\max } | T_{\max } \leq \Theta(N), T_{\mathrm{BPCC}}>t^{*}\right]\nonumber\\ & 
    \times \operatorname{Pr}\left[T_{\mathrm{BPCC}}>t^{*}\right]\end{align}
where $\mathbb{E}\left[T_{\max } | T_{\max } \leq \Theta(N), T_{\mathrm{BPCC}}>t^{*}\right]$ can be computed by
\begin{equation}
\begin{aligned}
&\mathbb{E}[T_{\max } | T_{\max } \leq \Theta(N), T_{\mathrm{BPCC}}>t^{*}] \\
=&~\frac{1}{\operatorname{Pr}\left[T_{\max } \leq \Theta(N), T_{\mathrm{BPCC}}>t^{*}\right]} \\ & 
\times \int_{t_{1}=0}^{\Theta(N)} \int_{t_{2}=t^{*}}^{\infty} t_{1} d \operatorname{Pr}\left[T_{\max } \leq t_{1}, T_{\mathrm{BPCC}} \leq t_{2}\right] \\ 
 \leq&~ \frac{\Theta(N)}{\operatorname{Pr}\left[T_{\max } \leq \Theta(N), T_{\mathrm{BPCC}}>t^{*}\right]} \\ & 
 \times \int_{t_{1}=0}^{\Theta(N)} \int_{t_{2}=t^{*}}^{\infty} d \operatorname{Pr}\left[T_{\max } \leq t_{1}, T_{\mathrm{BPCC}} \leq t_{2}\right] \\ 
=&~\Theta(N)
\end{aligned}\nonumber
\end{equation}
Since the master node receives at least $r$ rows of inner product results by time $T_{\mathrm{BPCC}}$, we have $S^{*}(T_{\mathrm{BPCC}})\geq r$. Next, since $S^*(t)$ is a monotonically increasing function with respect to time $t$, we can derive that, if $S^{*}(t^*)< r$, then $T_{\mathrm{BPCC}} > t^*$, which leads to
\begin{equation}\operatorname{Pr}\left[T_{\mathrm{BPCC}}>t^{*}\right] \leq \operatorname{Pr}\left[S^{*}\left(t^{*}\right)<r\right]=o\left(\frac{1}{N}\right)\nonumber
\end{equation}
Therefore, \equref{eq:secondTerm} can be written as
\begin{eqnarray}
{\mathbb{E}\left[T_{\mathrm{BPCC}} | \mathcal{E}_{1}^{c} \cap \mathcal{E}_{2}\right] \operatorname{Pr}\left[\mathcal{E}_{1}^{c} \cap \mathcal{E}_{2}\right]} &\leq&  {\Theta(N) \cdot o\left(\frac{1}{N}\right)}\nonumber\\ 
&=& o(1) \label{eq:second_term}
\end{eqnarray}

The third term in the right hand side of  \equref{eq:eqtwoevents} can be written as:
\begin{align} \label{eq:real_thirdTerm}&\mathbb{E}[T_{\mathrm{BPCC}} | \mathcal{E}_{1}^{c} \cap \mathcal{E}_{2}^{c}] \operatorname{Pr}\left[\mathcal{E}_{1}^{c} \cap \mathcal{E}_{2}^{c}\right] \nonumber\\=&~ \mathbb{E}\left[T_{\mathrm{BPCC}} | T_{\max } \leq \Theta(N), T_{\mathrm{BPCC}} \leq t^{*}\right] \nonumber\nonumber\\ & 
\times \operatorname{Pr}\left[T_{\max } \leq \Theta(N), T_{\mathrm{BPCC}} \leq t^{*}\right] \nonumber\\ \leq &~ \mathbb{E}\left[T_{\mathrm{BPCC}} | T_{\max } \leq \Theta(N), T_{\mathrm{BPCC}} \leq t^{*}\right]\nonumber\\
=&~\frac{1}{\operatorname{Pr}\left[T_{\max } \leq \Theta(N), T_{\mathrm{BPCC}}\leq t^{*}\right]} \nonumber\\ &
\times \int_{t_{1}=0}^{\Theta(N)} \int_{t_2=0}^{t^{*}} t_{2} d \operatorname{Pr}\left[T_{\max } \leq t_{1}, T_{\mathrm{BPCC}} \leq t_{2}\right] \nonumber\\ \leq &~\frac{t^*}{\operatorname{Pr}\left[T_{\max } \leq \Theta(N), T_{\mathrm{BPCC}}\leq t^{*}\right]} \nonumber\\ & 
\times \int_{t_{1}=0}^{\Theta(N)} \int_{t_2=0}^{t^{*}} d \operatorname{Pr}\left[T_{\max } \leq t_{1}, T_{\mathrm{BPCC}} \leq t_{2}\right] \nonumber\\ =&~t^*
\end{align}
Combining \equref{eq:eqtwoevents}, \equref{eq:thirdterm2}, \equref{eq:second_term}, and \equref{eq:real_thirdTerm},   we then have $\mathbb{E}\left[T_{\mathrm{BPCC}}\right] \leq t^{*}+o(1)$.


\subsection*{Proof of Theorem \ref{thm:tau}}
According to \lemmaref{lemma:lm1} and \theoref{thm:1}, we can derive
 \begin{equation}
    \tau^*-o(1)  \leq\mathbb{E}\left[T_{\mathrm{BPCC}}\right] \leq \tau^*+o(1) \nonumber
 \end{equation}
Therefore, 
$\lim _{N \rightarrow \infty} \mathbb{E}\left[T_{\mathrm{BPCC}}\right]=\tau^{*}$.

\subsection*{Proof of Theorem \ref{thm:convergence}}
Before showing the proof of Theorem \ref{thm:convergence}, we first present the following lemma, which will be used to prove this theorem. 

\begin{lemma}
\label{lemma:concave} 
Suppose $g(x)$ is a non-decreasing concave function and $g(x) \geq 0$, then 
\begin{equation}
    \frac{1}{p} \sum_{k=1}^{p} g\left(\frac{k}{p}\right) \geq \frac{1}{p+1} \sum_{k=1}^{p+1} g \left(\frac{k}{p+1}\right),
\end{equation}
for any $p \in \mathbb{Z}^+$.
\end{lemma}
\begin{proof}

Let $y_{k}= \frac{k}{p}$ and $z_{k}= \frac{k}{p+1}$, where $k\in [p]$. 
We then have 
\begin{equation}
    y_{k}=\left(1-\frac{k}{p}\right) z_{k}+\frac{k}{p} z_{k+1}. \ \  
\end{equation}
As $g(x)$ is a concave function, we have
\begin{align}
    g\left(y_{k}\right) &= g\left[\left(1-\frac{k}{p}\right) z_{k}+\frac{k}{p} z_{k+1}\right]\nonumber\\ &\geq\left(1-\frac{k}{p}\right) g\left(z_{k}\right)+\frac{k}{p} g\left(z_{k+1}\right)
\end{align}
Moreover, as $g(x)$ is also a non-decreasing function, which implies $g(z_k)\leq g(z_{k+1})$, we then have

\begin{align}
     g\left(y_{k}\right) &\geq \left(1-\frac{k}{p}\right) g\left(z_{k}\right)+\frac{k}{p} g\left(z_{k+1}\right)\nonumber\\
     &=  g\left(z_{k}\right)+\frac{k}{p}\left( g\left(z_{k+1}\right)-g\left(z_{k}\right)\right) \nonumber\\
     &\geq g\left(z_{k}\right)+\frac{k}{p+1}\left( g\left(z_{k+1}\right)-g\left(z_{k}\right)\right)\nonumber\\
     &= \left(1-\frac{k}{p+1}\right) g\left(z_{k}\right)+\frac{k}{p+1} g\left(z_{k+1}\right)
  \end{align}  
By summing  over all possible values of $k$ for both sides of the above inequality, we get
\begin{align}
     \sum_{k=1}^{p}  g\left(y_{k}\right) \geq& \frac{1}{p+1} \sum_{k=1}^{p}(p+1-k) g\left(z_{k}\right)\nonumber\\ & +\frac{1}{p+1} \sum_{k=1}^{p} k g\left(z_{k+1}\right)\nonumber\\
     =& \frac{1}{p+1} [ p \sum_{k=1}^{p} g\left(z_{k}\right)+\sum_{k=1}^{p}(1-k) g\left(z_{k}\right)\nonumber\\&+p g\left(z_{p+1}\right)+\sum_{k=2}^{p}(k-1) g\left(z_{k}\right) ]\nonumber\\
    =& \frac{p}{p+1} \sum_{k=1}^{p+1}  g\left(z_{k}\right)
\end{align}
which leads to $\frac{1}{p} \sum_{k=1}^{p} g\left(\frac{k}{p}\right) \geq \frac{1}{p+1} \sum_{k=1}^{p+1} g \left(\frac{k}{p+1}\right)$.
\end{proof}
\medskip

Now let's prove Theorem \ref{thm:convergence}. To prove that $\tau^* = \frac{r}{\beta}$ decreases with the increase of any $p_i$, $i\in [N]$, we just need to prove that $\beta$ increases with the increase of any $p_i$. Note that
\begin{equation}
\label{eq:beta1}
    \displaystyle \beta = \sum_{i=1}^{N}\frac{1}{\lambda_{i}}\left(1-\frac{1}{p_i}\sum_{k=1}^{p_i}e^{-\mu_i(\frac{\lambda_ip_i}{k}-\alpha_i)}\right)
\end{equation}
is dependent on both $\lambda_i$ and $p_i$. In the following, we first show that the change of $\lambda_i$ does not impact $\beta$. 
Particularly, by taking the partial derivative of $\beta$ with respect to $\lambda_i$, we have 
\begin{align}
\frac{\partial \beta}{\partial \lambda_i}=&-\frac{1}{\lambda_i^2}\left(1-\frac{1}{p_i}\sum_{k=1}^{p_i}e^{-\mu_i(\frac{\lambda_ip_i}{k}-\alpha_i)}\right)\nonumber\\&+\frac{1}{\lambda_i}\sum_{k=1}^{p_i}\frac{\mu_i}{k}e^{-\mu_i(\frac{\lambda_ip_i}{k}-\alpha_i)}\nonumber\\
=&-\frac{1}{\lambda_i^2}+\frac{1}{\lambda_i^2}\sum_{k=1}^{p_i}\left(\frac{1}{p_i}+\lambda_i\frac{\mu_i}{k}\right)e^{-\mu_i(\frac{\lambda_ip_i}{k}-\alpha_i)}\nonumber\\
\stackrel{(a)}{=}&-\frac{1}{\lambda_i^2}+\frac{1}{\lambda_i^2}\nonumber\\
=&0,
\end{align}

where $(a)$ is obtained  using \equref{eq:lambda}. Therefore, $\beta$ does not change as $\lambda_i$ varies. 

Next, we prove that $\beta$ increases with the increase of any $p_i$, $i\in[N]$.  
Define the following auxiliary function for each $i \in [N]$,
\begin{equation}
    g_i(x)=e^{-\frac{\mu_i\lambda_i}{x}}
\end{equation}
We can easily verify that $g_i(x)$ is a concave and increasing function, as $\mu_i>0, \lambda_i>0$, and $g_i(x)\geq0$. According to \lemmaref{lemma:concave},  $\frac{1}{p_i} \sum_{k=1}^{p_i} g_i\left(\frac{k}{p_i}\right) \geq \frac{1}{p_i+1} \sum_{k=1}^{p_i+1} g_i \left(\frac{k}{p_i+1}\right)$. Therefore, with the increase of any $p_i$, the term $\frac{1}{p_i}\sum_{k=1}^{p_i}e^{-\mu_i(\frac{\lambda_ip_i}{k}-\alpha_i)}$ in \equref{eq:beta1} decreases, causing $\beta$ to increase.

\subsection*{Proof of \theoref{thm:4}}
According to Theorem \ref{thm:convergence}, $\tau^*$ decreases with the increase of any $p_i$, $i\in [N]$. Therefore, the infimum of $\tau^*$ is attained when $p_i \rightarrow \infty$, $\forall i \in[N]$, i.e., 
\begin{equation}
    \inf \tau^* = \lim_{p_i\rightarrow \infty, \forall i \in[N]} \tau^* = \lim_{p_i\rightarrow \infty, \forall i \in[N]} \frac{r}{\beta}
\end{equation}
Let's now calculate $\lim_{p_i\rightarrow \infty, \forall i \in[N]} \beta$. According to \lemmaref{lemma:lambda}, we have $\lim_{p_i \rightarrow \infty}\lambda_i=\alpha_i$, which leads to 
\begin{align}
\label{eq:beta_converge}
    &\lim_{p_i\rightarrow\infty, \forall i \in [N]}\beta \nonumber\\=&\lim_{p_i\rightarrow\infty, \forall i \in [N]}\sum_{i=1}^{N}\frac{1}{\lambda_{i}}\left(1-\frac{1}{p_i}\sum_{k=1}^{p_i}e^{-\mu_i(\frac{\lambda_ip_i}{k}-\alpha_i)}\right) \nonumber\\
    =& \lim_{p_i\rightarrow\infty, \forall i \in [N]}\sum_{i=1}^{N}\frac{1}{\lambda_{i}}\left(1-e^{\mu_i\alpha_i}\int_{0}^{1}e^{-\frac{\mu_i\lambda_i}{x}}dx \right) \nonumber \\
    =&\sum_{i=1}^N\frac{1}{\alpha_i}(1-e^{\mu_i\alpha_i}\int_{0}^{1}e^{-\frac{\mu_i\alpha_i}{x}}dx)
\end{align}
Therefore, 
\begin{align}
    \inf \tau^*= &\lim_{p_i \rightarrow \infty, \forall i\in [N]} \tau^*\nonumber\\ = &  \frac{r}{\sum_{i=1}^N\frac{1}{\alpha_i}(1-e^{\mu_i\alpha_i}\int_{0}^{1}e^{-\frac{\mu_i\alpha_i}{x}}dx)}
\end{align}

Similarly, from Theorem \ref{thm:convergence}, we can derive that the supremum of $\tau^*$ is attained when $p_i=1$, $\forall i\in[N]$, i.e., 
\begin{equation}
    \sup \tau^* =  \sum_{i=1}^{N}\frac{1}{\sup \lambda_{i}}\left(1-e^{-\mu_i(\sup \lambda_i-\alpha_i)}\right),
\end{equation}
where $\sup \lambda_i$ is given by \equref{eq:lambda_sup}. 

\subsection*{Proof of Corollary \ref{corollary:2}}
As  $\tau^*$ approaches its infimum when $p_i \rightarrow \infty$, $\forall i\in[N]$, we have 
\begin{align}
    \hat{\ell_i}=&\lim_{p_j\rightarrow \infty, \forall j \in [N]}\ell^*_i\nonumber\\ =&\lim_{p_j\rightarrow \infty, \forall j \in [N]} \frac{r}{\beta\lambda_i}\nonumber\\ =& \frac{r}{\alpha_i\sum_{j=1}^N\frac{1}{\alpha_j}(1-e^{\mu_j\alpha_j}\int_{0}^{1}e^{-\frac{\mu_j\alpha_j}{x}}dx)}, 
\end{align}
according to Eq.~\eqref{eq:prob1_solution} and 
\equref{eq:beta_converge} in the proof of Theorem \ref{thm:4}.

\subsection*{Proof of Theorem \ref{thm:3}}

According to Theorem \ref{thm:tau}, we have $\lim _{N \rightarrow \infty} \mathbb{E}\left[T_{\mathrm{BPCC}}\right]=\tau^{*}$. Similarly, according to \cite{reisizadeh19coded}, we can derive $\lim _{N \rightarrow \infty} \mathbb{E}\left[T_{\mathrm{HCMM}}\right]=\tau^*_H$. Since HCMM is a special case of BPCC with $p_i = 1$, $\forall i \in [N]$, 
by applying \theoref{thm:convergence}, 
we have 
 \begin{equation}
         \lim _{N \rightarrow \infty} \mathbb{E}\left[T_{\mathrm{BPCC}}\right]\leq\lim _{N \rightarrow \infty} \mathbb{E}\left[T_{\mathrm{HCMM}}\right] \nonumber
 \end{equation}

\end{document}